\documentclass[12pt]{iopart}

\usepackage{iopams}  

\expandafter\let\csname equation*\endcsname\relax
\expandafter\let\csname endequation*\endcsname\relax
\usepackage{amsmath}
\usepackage{hyperref}

\usepackage{comment}
\usepackage{bbm}
\usepackage{tensor}
\usepackage{mathtools}
\usepackage{xcolor}						
\usepackage{etoolbox}
\usepackage{enumitem}
\usepackage{accents}

\usepackage{etoolbox}
\makeatletter
\newcommand{\mainmatter}{%
  \setcounter{footnote}{0}%
  \patchcmd{\@makefntext}{\fnsymbol}{\arabic}{}{}%
  \patchcmd{\@thefnmark}{\fnsymbol}{\arabic}{}{}%
  \def\@makefnmark{\textsuperscript{\arabic{footnote}}}%
}
\makeatother

\newcommand{\II}{\mathbb{I}}

\renewcommand{\Im}{\mathrm{{Im}}}

\newcommand{\Pj}[2]{P^{#1}_{#2}}   
\newcommand{\Pji}[2]{(P^{-1})^{#1}_{#2}}

\newcommand{\SLtwoC}{\mathrm{SL}(2,\mathbb{C})}	
\newcommand{\sltwoC}{\mathfrak{sl}(2,\mathbb{C})}

\newcommand{\LGC}{\mathrm{SO}(3,1;\mathbb{C})}
\newcommand{\LGCp}{\mathrm{SO}(3,1)^+_\mathbb{C}}
\newcommand{\la}{\mathfrak{so}(3,1)}
\newcommand{\laC}{\mathfrak{so}(3,1)_\mathbb{C}}
\newcommand{\laCp}{\mathfrak{so}(3,1)^+_\mathbb{C}}
\newcommand{\laCm}{\mathfrak{so}(3,1)^-_\mathbb{C}}
\newcommand{\laCpm}{\mathfrak{so}(3,1)^\pm_\mathbb{C}}
\newcommand{\RR}{\mathbb{R}}
\newcommand{\CC}{\mathbb{C}}
\newcommand{\MM}{\mathbb{M}}

\newcommand{\Pud}[2]{\tensor{P}{^{#1}_{#2}}}

\newcommand{\Ppm}[2]{\prescript{\pm}{}{\tensor{P}{^{#1}_{#2}}}}
\newcommand{\Ppmi}[2]{\prescript{\pm}{}{\tensor{P}{_{#1}^{#2}}}}
\newcommand{\Pmp}[2]{\prescript{\mp}{}{\tensor{P}{^{#1}_{#2}}}}
\newcommand{\Pp}[2]{\prescript{+}{}{\tensor{P}{^{#1}_{#2}}}}
\newcommand{\Pm}[2]{\prescript{-}{}{\tensor{P}{^{#1}_{#2}}}}
\newcommand{\Pjpm}[2]{{\prescript{\pm}{}{P}^{#1}_{#2}}}
\newcommand{\Pjpmi}[2]{{(\prescript{\pm}{}{P^{-1}})^{#1}_{#2}}}

\newcommand{\D}[1]{\mathrm{d}#1}
\newcommand{\T}[2]{\tensor{#1}{#2}}

\newcommand{\lc}[1]{\tensor{\epsilon}{#1}}
\newcommand{\sce}[1]{\tensor{\mathring{\epsilon}}{#1}}

\newcommand{\scf}[1]{\tensor{\mathring{f}}{#1}}
\newcommand{\scfpm}[1]{\tensor{\mathring{f}}{^\pm#1}}

\newcommand{\lcdp}[1]{\tensor{\tilde{\epsilon}}{#1}}    
\newcommand{\lcdm}[1]{\tensor{\underaccent{\tilde}{\epsilon}}{#1}} 

\newcommand{\tE}{\tilde{E}}
\newcommand{\Et}{\underaccent{\tilde}{E}}

\newcommand{\Nt}{\underaccent{\tilde}{N}}
\newcommand{\tN}{\tilde{N}}

\newcommand{\grad}[1]{\mathrm{d}{#1}}

\newcommand{\curlyD}{\mathcal{D}}

\newcommand{\curlyL}{\mathcal{L}}
\newcommand{\curlyM}{\mathcal{M}}

\newcommand{\rb}[1]{\left({#1}\right)}
\newcommand{\brb}[1]{\bigl({#1}\bigr)}
\newcommand{\Brb}[1]{\Bigl({#1}\Bigr)}

\newcommand{\curlybrackets}[1]{\left\{{#1}\right\}}

\newcommand{\SUtwo}{\mathrm{SU}(2)}			
\newcommand{\sutwo}{\mathfrak{su}(2)}			
\newcommand{\sothreeone}{\mathrm{so}(3,1)} 
\newcommand{\g}{\mathfrak{g}}        

\newcommand{\dxd}[2]{\!\mathrm{d}^{#1}{#2\,}}				

\newcommand{\commutator}[2]{\left[{#1},{#2}\right]}			
\renewcommand{\PB}[2]{\left\{#1,#2\right\}}	

\newcommand{\slotdot}{\boldsymbol{\,\cdot\,}}

\newcommand{\addref}{{\color{teal} add reference}}

\usepackage{amsthm}
\newtheorem{theorem}{Theorem}[section]  
\newtheorem{proposition}[theorem]{Proposition}

\begin{document}

\title[Revisiting LQG with selfdual variables: Classical theory]{Revisiting loop quantum gravity with selfdual variables: Classical theory}

\author{Hanno Sahlmann$^1$ and Robert Seeger$^2$}

\address{Friedrich-Alexander-Universit\"at Erlangen-N\"urnberg (FAU)\\ Institute for Quantum Gravity, Staudtstra{\ss}e 7/B2, 91058 Erlangen, Germany}
\ead{$^1$ hanno.sahlmann@gravity.fau.de, $^2$ robert.seeger@gravity.fau.de}
\vspace{10pt}
\begin{indented}
\item[]October 2023
\end{indented}

\begin{abstract}
We review the classical formulation of general relativity as an $\SLtwoC$ gauge theory in terms of Ashtekar's selfdual variables and reality conditions for the spatial metric (RCI) and its evolution (RCII), and we add some new observations and results. 

We first explain in detail how a connection taking values in the Lie algebra of the complex Lorentz group yields two Spin(3,1) connections, one selfdual and one anti-selfdual, without the need for a spin structure. 

We then demonstrate that the selfdual part of the complexified Palatini action in Ashtekar variables requires a holomorphic phase space description in order to obtain a non-degenerate symplectic structure. The holomorphic phase space does not allow for the implementation of the reality conditions as additional constraints, so they have to be taken care of ``by hand'' during the quantisation. We also observe that, due to the action being complex, there is an overall complex phase that can be chosen at will.   

We then review the canonical formulation and the consequences of the implementation of the reality conditions. We pay close attention to the transformation behaviour of the various fields under (complex) basis changes, as well as to complex analytic properties of the relevant functions on phase space. 
\end{abstract}

%
%
%
%
%
\mainmatter
\section{Introduction}\label{sec_introduction}

If one wants to describe fermions coupled to gravity, the metric field does not suffice. Rather, one needs to introduce frame fields, and with them the symmetry under frame rotations, a spin structure and hence a local $\SLtwoC$ symmetry. In this context, as the fermions in the standard model are chiral, it is of special interest to consider chiral formulations of gravity. One example are the variables introduced by Sen \cite{sen_gravity_1982} and Ashtekar \cite{ashtekar_new_1986,ashtekar_new_1987,ashtekar_gravitational_2021} for gravity. The basic variables are a frame field $e^I$ and an $\laC$ connection $\omega$. The fields are selfdual under the Hodge operator on $\laC$, 
\begin{equation}
    *\omega = \pm i \omega, \qquad *(e\wedge e)= \pm i (e\wedge e). 
\end{equation}
As a consequence, these variables are complex valued, and the relation to real-valued gravitational fields is not entirely straightforward. On the other hand, these variables have very interesting properties in a canonical formulation of gravity: Pulling $e\wedge e$ and $\omega$ back to a spatial slice yields canonical variables and a first order system of constraints \cite{ashtekar_new_1987}. No further gauge fixing is necessary. Moreover, the structure of the constraints -- in particular that of the Hamilton constraint -- is rather simple. Finally, when coupling to fermions, it is the selfdual connection $\omega$ that couples to the right-handed fermions, and its complex conjugate to the left-handed ones. In the supersymmetric context, it is even more apparent that the chiral theory is special. It retains some manifest supersymmetry in the canonical theory \cite{eder_super_2021,eder_holst-macdowell-mansouri_2021}. 

For these reasons, Ashtekar variables seem to be an excellent starting point for a quantization program for gravity using techniques from gauge theory. Consequently there has been considerable work done in this direction, from early works  \cite{thiemann_solution_1993,kastrup_spherically_1994,ashtekar_coherent_1994,thiemann_reality_1996} over later developments  \cite{wieland_complex_2012,wieland_twistorial_2012,wilson-ewing_loop_2015,achour_analytic_2015,achour_loop_2015} to recent works \cite{varadarajan_euclidean_2019,ashtekar_gravitational_2021, alexander_inner_2022,wieland_simplicial_2023}.

Nevertheless most progress in loop quantum gravity (see for example \cite{thiemann_modern_2007,ashtekar_background_2004}) has been made using a formulation in terms of a \emph{real valued} $\SUtwo$ connection, the Ashtekar-Barbero connection \cite{barbero_g_real_1995,immirzi_real_1997}. The reasons for this are technical in nature: Compactness of $\SUtwo$ simplifies the construction of measures on spaces of connections. Furthermore, the reality conditions hard to implement in the quantum theory. 

There are some hints from the $\SUtwo$ theory that a more fundamental theory with $\SLtwoC$ gauge symmetry might exist. The $\SUtwo$ theory would then be a technically simpler theory that has, at least in some cases, a precise relationship with the fundamental $\SLtwoC$ theory. This situation would be in close analogy to that in quantum field theory in which Euclidean and Lorentzian theory are related in non-trivial ways by the Wick transform. For example, \cite{thiemann_canonical_1993,wilson-ewing_loop_2015,achour_loop_2015,ben_achour_spherically_2017} demonstrate that Ashtekar variables can be used for quantization in a symmetry-reduced context. The same is true even for the supersymmetric theory \cite{eder_supersymmetric_2021}. For the calculation of black hole entropy from loop quantum gravity using real variables, the counting of states of an $\SUtwo$ Chern-Simons theory is essential, and the result depends on the Barbero-Immirzi parameter $\beta$. It has been observed that a certain analytic continuation, corresponding to the limit $\beta \longrightarrow i$, i.e., to a change to the chiral theory, and to an $\SLtwoC$ Chern-Simons theory leads to the  Bekenstein-Hawking entropy without any free parameters \cite{frodden_black-hole_2014,jibril_analytic_2015}. Again, it turns out that an analogous result holds in the supersymmetric theory \cite{eder_towards_2022,eder_chiral_2022}. 

Given the evidence that there could exist a chiral quantum theory that is closely related to the $\SUtwo$ theory, it is a natural task to look for this theory more directly. The present article aims to contribute towards this goal. We lay out the classical theory as clearly, precisely and as general as we can and explain the options one has as one proceeds to quantize. 
Some of this general ground has been covered elsewhere, see for example \cite{ashtekar_self_1992,baez_gauge_1994,thiemann_reality_1996,ashtekar_coherent_1996,ashtekar_lectures_1998,wieland_complex_2012,ashtekar_gravitational_2021}.

We hope to add to the classical foundations of the subject in several ways, however. the most important ones are: 
\begin{itemize}
    \item We analyze the structure of the variables of complex GR and demonstrate in detail how the existence of the selfdual and anti-selfdual connections follow. In particular, we review the structure of $\LGC$ and use that to show how the $\laC$-connection of complex Palatini gravity can be transformed into the selfdual and anti-selfdual $\sltwoC$-valued Ashtekar connections. 
    \item We show in detail that complexified GR only leads to a non-degenerate symplectic structure if a holomorphic viewpoint is taken.\footnote{Another starting point is a non-holomorphic phase space not derived from an action. One can then formulate constraints that lead to classical GR when implemented. This is the starting point taken in \cite{wieland_complex_2012}.}    
    \item In our calculation, we strive to carefully distinguish internal three- and four-dimensional metrics, structure constants and Cartan-Killing metrics, to bring out their relationships. Since we make our choices of bases explicit, it is also clear how all objects would transform under changes of bases\footnote{In particular the choice of basis in $\sltwoC$ can lead to confusion as it influences the reality of various functions.}. 
    \item We observe that there is an undetermined phase in front of the action of complex GR that influences the symplectic structure and the constraints, but not the dynamics. We show that the reality conditions are not affected by this phase, and will always lead to real Einstein gravity, however with a changed symplectic structure. It appears that the resulting symplectic structure has been also used in \cite{thiemann_reality_1996,thiemann_account_1995}. This change may have consequences for the quantum theory.   
\end{itemize}

This article consists of two major parts. 
In section \ref{sec:selfdual_GR} we consider the construction of a holomorphic formulation of general relativity in terms of Ashtekar's selfdual variables. 
Section \ref{sec:reality_conditions} then deals with the formulation of the reality condition, which are necessary in order to retrieve a description of real gravity.
After an introduction to complexified GR in section \ref{sec:complex_GR}, in section \ref{sec:complex_Lorentz_algebra_selfduality_SLtwoC} we describes in detail how one obtains Ashtekar's $\SLtwoC$ connection from a complexified Lorentz connection, under a split of this connection into its selfdual and anti-selfdual parts. In section \ref{sec:SelfdualPalatiniAction} we use this to split the Palatini action into its selfdual and anti-selfdual contributions, from which we construct a holomorphic Hamilton formulation of the selfdual action. There, we introduces a phase in front of the action and discuss the consequences for the Poisson relation of the canonical pair and the constraints. 
In section \ref{sec:RCI} we are concerned with the recovery of a real spatial metric from selfdual GR, i.\,e. we describe the first reality condition and its formulation in terms of Ashtekar's variables. In section \ref{sec:RCII} we describe the second reality condition, that is supposed to keep the spatial metric real with respect to the dynamics. 
As a justification of the reality conditions we consider their classical implementation in order to recover real ADM gravity in section \ref{sec:recovering_ADM}. 
Finally, in section \ref{sec:discussion} we discuss what understanding we were able to add to the description of GR in selfdual variables and give and outlook to a possible quantisation.
Regarding conventions, we use the metric signature $(-,+,+,+)$.

\section{Selfdual holomorphic general relativity in Ashtekar variables}\label{sec:selfdual_GR}
In this section we present the formulation of general relativity in Ashtekar variables. To this end we start with complex general relativity, based on the notion of a complex Lorentz connection and a complex tetrad. The underlying structure group is thus the complex Lorentz group $\LGC$. We then single out the selfdual parts and derive the corresponding Hamilton formulation. For other accounts of this see \cite{ashtekar_new_1986,baez_gauge_1994,ashtekar_lectures_1998,ashtekar_gravitational_2021}.

\subsection{Complexified general relativity}\label{sec:complex_GR}
We start with the description of complexified general relativity, as a theory of gravity over a real differential manifold $\curlyM$, but with a complexified tangent bundle $T_\CC\curlyM$. In this setup coordinates are real but all objects related to tangent space are understood to be complex. 
This setup is described in detail or example in \cite{baez_gauge_1994}.

In the Palatini formulation of complex general relativity, we therefore work with complex frame fields $e^I_\alpha$ and a complex Lorentz connection $\T{\omega}{^I_J}=\T{\omega}{_\alpha^I_J}\grad{x^\alpha}$ with
\begin{equation}\label{eq:lorentz_algebra}
     \T{\omega}{_{\alpha IJ}}=\T{\omega}{_\alpha^K_J}\eta_{KI}=-\T{\omega}{_{\alpha JI}}.
\end{equation}
Here and in the following, Lorentz indices are pulled with $eta$. This connection is understood as a connection of the complexified Lorentz group $\LGC$, which we will look at in detail in section.

The frame fields act as soldering forms between the complexified tangent space and the complexified internal Minkowski space which is consequently also complexified. The tetrads $e^I_\mu$ relate the complex spacetime metric $g_{\alpha\beta}$ to the Minkowski metric $\eta_{IJ}$: 
\begin{equation}
    g_{\alpha\beta}=e^I_\alpha e^J_\beta \eta_{IJ}. 
\end{equation}
Note, however, that the Minkowski metric $\eta$ is not complex, but stays real.

In this complexified version of general relativity, we consequently analyse a complex version of the Palatini action
\begin{equation}
    S_\CC[\omega,e]=\frac{\lambda}{\kappa}\int_\curlyM \epsilon_{IJKL}\Sigma^{IJ}\!\wedge F^{KL}.
\end{equation}
This is constructed from the complex two-form
\begin{equation}
    \Sigma^{IJ}:=e^I\wedge e^J = e^I_\alpha e^J_\beta \mathop{\D{x^\alpha}\wedge\D{x^\beta}} 
\end{equation}
and the curvature of the complex Lorentz connection 
\begin{equation}\label{eq:curvature_LorentzConnection}
    \T{F}{^I_J}[\omega]=\grad{\T{\omega}{^I_J}}+\T{\omega}{^I_K}\wedge\T{\omega}{^K_J}.
\end{equation}
Note that we rescale the action by a factor $\lambda\in\CC$. This has, in the end, to do with the quantum theory corresponding to this action. We comment on its role in this classical description later, where its effects are visible and discuss the specific cases $\lambda=1,i$.

Ultimately, we want to restrict to the selfdual part of this complex action. This is done in section \ref{sec:SelfdualPalatiniAction}. Before this, however, we recall the complexification of the Lorentz Lie algebra and its relation to $\sltwoC$ in the next section.
   
\subsection{Complexification of the Lorentz Lie algebra, selfduality and $\sltwoC$}
\label{sec:complex_Lorentz_algebra_selfduality_SLtwoC}

In the following we have occasion to use complexifications of Lie algebras. We summarised the important aspects in \ref{app:complexification_lie_algebras}.

\subsubsection{(Anti-)selfdual split of the complexified Lorentz algebra}
In signature $(-+++)$, the Hodge star $*$ is a skew involution, $*^2=-1$. As a map $\Lambda^2(\MM)\rightarrow \Lambda^2(\MM)$ it is selfadjoint with respect to the natural inner product. To obtain a decompoisition into eigenspaces, we have to work with the complexification $\Lambda^2(\MM)_\CC$.  On this space $*$ has eigenvalues $\pm i$ and the projectors onto the eigenspaces are 
\begin{equation}
    \Ppm{}{}:=\frac{1\mp i\ast}{2}.
\end{equation} 
Working in a basis, we get the explicit expression
\begin{equation}
\Ppm{IJ}{KL}:=\frac{1}{2}\rb{\delta^{[I}_K\delta^{J]}_L\mp\frac{i}{2}{\epsilon^{IJ}}_{KL}},
\end{equation}
acting on an object with two internal upper indices.
From this we directly see that the projectors are symmetric under exchanging the first and second pair of indices, i.\,e. 
\begin{equation}
    \Ppm{IJ}{KL}=\Ppmi{KL}{IJ}.
\end{equation}
Hence we have immediately the action on internal two-forms.
It is easy to see that $\Ppm{}{}$ is selfdual/anti-selfdual in both pairs of indices: 
\begin{equation}
\begin{aligned}
    \frac{1}{2}\lc{^I^J_M_N}\Ppm{MN}{KL}=\pm i \Ppm{IJ}{KL},\\
     \frac{1}{2}\lc{_K_L^M^N}\Ppm{IJ}{MN}=\pm i \Ppm{IJ}{KL}.
\end{aligned}     
\end{equation}
Since $\Ppm{}{}$ are projections onto eigenspaces, one has
\begin{equation}\label{eq:projection_properties}
\begin{aligned}
&\Ppm{IJ}{KL}\Pmp{KL}{MN}=0,\\ 
&\Ppm{IJ}{KL}\Ppm{KL}{MN}=\Ppm{IJ}{MN},\\ &\Pp{IJ}{KL}+\Pm{IJ}{KL}=\delta^{[I}_K\delta^{J]}_L.
\end{aligned}
\end{equation}
So $\Pp{}{}$ and $\Pm{}{}$ are indeed orthogonal projectors. 

In the following we will use the basis
\begin{equation}\label{eq:generators_laC}
    \T{(L_{IJ})}{^K_L}=\delta^K_{[I}\eta_{J]L}
\end{equation}
of $\laC$.
In this basis, the components for a Lie algebra  element, e.\,g. the connection $\omega^I{}_J$, are simply $\omega^{IJ}$.

Given this basis, we can now understand how the notion of (anti-)selfduality, as a concept on the internal Minkowski space, translates to $\laC$ matrices. 
Acting on the labels of the generator yields 
\begin{equation}
    \Ppmi{IJ}{MN}\T{(L_{MN})}{^K_L}=\Ppmi{IJ}{MN}\delta^K_{[N}\eta_{M]L}=\T{\Ppm{}{}}{_{IJ}^K_L}.
\end{equation}
Because of the symmetry of labels and matrix indices of the generators, i.\,e. $\T{(L_{IJ})}{_{KL}}$, this is exactly the same as acting with the projection on the matrix indices:
\begin{equation}
    \T{\Ppm{}{}}{^K_{LM}^N}\T{(L_{IJ})}{^M_N}=\T{\Ppm{}{}}{_{IJ}^K_L}.
\end{equation}
Hence (anti-)selfdual projection of components makes the corresponding matrix (anti-)selfdual as well. As a consequence, the (anti-)selfdual generators are
\begin{equation}
    {L^{\pm}}_{IJ}:=\Ppm{MN}{IJ}L_{MN},
\end{equation}
with
\begin{equation}\label{eq:selfdual_generators_laCpm}
    \T{({L^{\pm}}_{IJ})}{^K_L}:=\T{\Ppm{}{}}{_{IJ}^K_L}.
\end{equation}

With this established, we can now use the last relation in \eref{eq:projection_properties} to split $a\in\laC$ into its selfdual and anti-selfdual parts 
\begin{equation}
    a=a^{IJ}L_{IJ}=a^{IJ}\rb{\Pp{KL}{IJ}+\Pm{KL}{IJ}}L^{KL}=a^{+IJ}{L^+}_{IJ}+a^{-IJ}{L^-}_{IJ},
\end{equation}
where both the components $a^{\pm IJ}:=\Ppm{IJ}{KL}a^{KL}$ and the generators ${L^{\pm}}_{IJ}$ are now (anti-)selfdual. As the projection $\Ppm{}{}$ is orthogonal, this decomposes $\laC$ into
selfdual and anti-selfdual parts $\laCpm:=\Ppm{}{}\laC$, respectively:
\begin{equation}
    \laC\cong\laC^+\oplus\laC^-.
\end{equation}
Elements of $\laCpm$ therefore are of the form 
\begin{equation}\label{eq:selfdual_alg_elements}
    a^\pm=a^{IJ}\Ppmi{IJ}{MN}L_{MN}=a^{\pm IJ}{L^{\pm}}_{IJ}.
\end{equation}

The next step is to  analyse the commutation relation of $\laC$, in order to show that $\laCpm$ are indeed Lie subalgebras. The commutation relation of $\laC$ is given by
\begin{equation}
    \commutator{L_{IJ}}{L_{KL}}=\scf{_{IJKL}^{MN}}L_{MN},
\end{equation}
where in the basis \eref{eq:generators_laC}, the structure constants $\scf{}{}$ are 
\begin{equation}
    \scf{_{IJKL}^{MN}}=2\delta^{[M}_{[I}\eta_{J][K}\delta^{N]}_{L]}.
\end{equation}
We now look at this commutation relation when $a,b\in\laC$ are decomposed into their (anti-)selfdual parts:
\begin{equation}
    \begin{aligned}
    \commutator{a}{b}&=\commutator{a^++a^-}{b^++b^-}
    =\commutator{a^+}{b^+}+\commutator{a^+}{b^-}+\commutator{a^-}{b^+}+\commutator{a^-}{b^-}.
    \end{aligned}
\end{equation}
In order to understand the action of the projection on this commutator, we have the to understand both the purely (anti-)selfdual commutators and the mixed ones.

To this end we consider the adjoint action on $\laC$. Let still $a,b\in\laC$:
\begin{equation}
    \pi(a)b:=\commutator{a}{b}.
\end{equation}
We know $\laCpm\subset\laC$ is invariant. Furthermore let $b^\pm\in\laCpm$. Hence,
\begin{equation}
    b\in\laCpm \Leftrightarrow \Ppm{}{}b=b.
\end{equation}
Now
\begin{equation}
    \begin{aligned}
    \pi(a)b^\pm\in\laCpm &\Leftrightarrow \Ppm{}{}\pi(a)b^\pm=\pi(a)b^\pm \\
    &\Leftrightarrow\Ppm{}{}\commutator{a}{b^\pm}=\commutator{a}{b^\pm}\quad\forall a \in \laC, \forall b^\pm\in \laCpm.
    \end{aligned}
\end{equation}

Restricting to $a^\pm\in\laCpm$, this tells us that the commutators of (anti-)selfdual Lie algebra elements are indeed (anti-)selfdual, even before the actual projection. 
Restricting instead to $a^\mp$, the commutators $\commutator{a^\mp}{b^\pm}$, with what we just showed, have to both selfdual and anti-selfdual. As the (anti-)selfdual decomposition is an orthogonal one, these commutators have to vanish. 
Hence we are left with 
\begin{equation}
    \commutator{a}{b}=\commutator{a^+}{b^+}+\commutator{a^-}{b^-}
\end{equation}
and projection yields 
\begin{equation}
\label{eq:P_is_a_homomorphism}
    \Ppm{}{}\commutator{a}{b}=\commutator{a^\pm}{b^\pm}.
\end{equation}
This in turn shows that $\laCpm$ are closed under the commutator and hence are Lie algebras on their own. 

The structure constants of these Lie subalgebras can be determined from the ones of $\laC$ as follows. As all (anti-)selfdual Lie algebra elements are still in $\laC$, it holds that 
\begin{equation}
    \commutator{a^\pm}{b^\pm}=a^{\pm IJ}b^{\pm KL}\scf{_{IJKL}^{MN}}L^\pm_{MN}.
\end{equation}
So at first glance, the structure constants of $\laCpm$ seem to be the same as the ones of $\laC$. This however disregards the fact that all pairs of indices of the structure constants are (anti-)selfdually projected. The actual structure constants of $\laCpm$ are therefore given by 
\begin{equation}
    \commutator{L^\pm_{IJ}}{L^\pm_{KL}}=\scfpm{_{IJKL}^{MN}}L^\pm_{MN}
\end{equation}
with 
\begin{equation}\label{eq:structure_constants_laC}
    \scfpm{_{IJKL}^{MN}}:=\Ppm{I'J'}{IJ}\Ppm{K'L'}{KL}\scf{_{I'J'K'L'}^{M'N'}}\Ppm{M'N'}{MN}.
\end{equation}
Using the explicit form of the structure constants in \eref{eq:structure_constants_laC}, the selfdual structure constants can be expressed solely in terms of contracted projections, namely
\begin{equation}
    \scfpm{_{IJKL}^{MN}}=-2\T{\Ppmi{}{}}{_{IJ}^{K'}_{L'}}\T{\Ppmi{}{}}{_{KL}^{L'}_{M'}}\T{\Ppm{}{}}{^{MNM'}_{K'}}=-2\tr\rb{\T{{L^\pm}}{_{IJ}}\T{{L^\pm}}{_{KL}}\T{{L^\pm}}{^{MN}}},
\end{equation}
where we employed the relation between projectors and the (anti-)selfdual generators.

Going back one step, we can look again at the projection of the $\laC$ commutator
\begin{equation}\label{eq:commutator_projection}
\begin{aligned}
    &\Ppm{}{}\commutator{a}{b}=\commutator{a^\pm}{b^\pm}=\commutator{\Ppm{} {}a}{\Ppm{}{}b}\\
    &\Leftrightarrow a^{IJ}b^{KL}\scf{_{IJKL}^{MN}}\Ppm{M'N'}{MN}L_{M'N'}=a^{I'J'}b^{K'L'}\Ppm{IJ}{I'J'}\Ppm{KL}{K'L'}\scf{_{IJKL}^{MN}}L_{MN}.
\end{aligned}    
\end{equation}
Here we treated the left hand side as the projection of the commutator and the right hand side as the commutator of the projections, explicitly.
 This now expresses the fact that the projectors are Lie algebra homomorphisms and that the adjoint representation of $\sltwoC_\CC$ is reducible. 
This yields the relation 
\begin{equation}\label{eq:selfdual_sc_relation}
    \scf{_{IJKL}^{M'N'}}\Ppm{MN}{M'N'}=\Ppm{I'J'}{IJ}\Ppm{K'L'}{KL}\scf{_{I'J'K'L'}^{MN}}.
\end{equation}
Projecting the pair $MN$ is trivial on the left hand side, as $(\Ppm{}{})^2=\Ppm{}{}$, but gives the (anti-)selfdual structure constants on the right hand side. Doing this and using the total antisymmetry of $\scf{_{IJKLMN}}$ and $\scfpm{_{IJKLMN}}$ in the pairs $IJ$, $KL$ and $MN$ yields
\begin{equation}\label{eq:structure_constants_laCpm}
    \scfpm{_{IJKLMN}}=\Ppm{I'J'}{IJ}\scf{_{I'J'KLMN}}=\Ppm{K'L'}{KL}\scf{_{IJK'L'MN}}=\Ppm{M'N'}{MN}\scf{_{IJKLM'N'}}.
\end{equation}
Hence it is sufficient to project just one of the pairs of indices of the $\laC$ structure constants. 
 
With the structure constants for both $\laC$ and $\laCpm$, we can compare the respective Cartan-Killing metrics. 
For $\laC$ it is given by 
\begin{equation}
    k_{IJKL}=2\tr\rb{L_{IJ}L_{KL}}=2\delta^M_{[I}\eta_{J]N}\delta^N_{[K}\eta_{L]M}=-2\eta_{I[K}\eta_{L]J}.
\end{equation}
The same result is obtained in a more tedious fashion, when determining the Cartan-Killing metric from contraction of two structure constants according to $k_{IJKL}=\scf{_{IJMN}^{M'N'}}\scf{_{KLM'N'}^{MN}}$. 

For $\laCpm$, however, we start from the structure constants and use \eref{eq:structure_constants_laCpm}:
\begin{equation}
    \T{{k^\pm}}{_{IJKL}}=\scfpm{_{IJMN}^{M'N'}}\scfpm{_{KLM'N'}^{MN}}=\Ppm{I'J'}{IJ}\Ppm{K'L'}{KL}\scf{_{I'J'MN}^{M'N'}}\scf{_{K'L'M'N'}^{MN}}.
\end{equation}
So the Cartan-Killing metric of $\laCpm$ is indeed the projection of the Cartan-Killing metric of $\laC$, i.\,e.
\begin{equation}\label{eq:selfdual_CKm_projector}
    \T{{k^\pm}}{_{IJKL}}=\Ppm{I'J'}{IJ}\Ppm{K'L'}{KL}k_{I'J'K'L'}=-2\T{\Ppm{}{}}{_{K'L'IJ}}\Ppm{K'L'}{KL}=-2\T{\Ppm{}{}}{_{IJKL}},
\end{equation}
where in the last step the symmetry under exchange of first and second pair of indices was used. Note that the pair $I,J$ here is pulled down with Minkowski metrics. So in combination with the factor of $-2$, we can understand this as 
\begin{equation}
    \T{{k^\pm}}{_{IJKL}}=\Ppm{MN}{IJ}k^\pm_{MNKL}=\Ppm{MN}{KL}k^\pm_{IJMN}
\end{equation}
and again, projecting one pair of indices is equivalent to projecting all pairs.

We want to make a remark about dimensions.
The projection to the (anti-)selfdual subalgebra halves the dimension of $\laC$ from $6\,\CC$ to $3\,\CC$ for $\laCpm$, respectively. The structure constants of $\laC$ are the same as for $\la$ and the projection does not change the symmetries of $\scfpm{}$ as compared to $\scf{}$.  Because of this and the $6$ real dimensions of the (anti-)selfdual algebra, we see that indeed both  $\laCp$ and $\laCm$ are isomorphic to $\la$ as real Lie algebras. 

We also remark that, besides $\sothreeone^{+}$ and $\sothreeone^{-}$, the (real) subspace of real matrices fulfilling \eref{eq:lorentz_algebra} furnishes a third subalgebra of $\sothreeone_\CC$. For $a\in \sothreeone_\CC$ one has
\begin{equation}
    \Im\; a = 0 \quad \Longleftrightarrow \quad a^- =\overline{a^+}. 
\end{equation}

\subsubsection{Relation to $\sltwoC$}\label{sec:Relation_sltwoC}

As we want to formulate selfdual general relativity in terms of an $\SLtwoC$ connection, we need go over from the (anti-)selfdual Lorentz algebra to $\sltwoC$. 
In order to fix the notation,
\begin{equation}
    \sltwoC=\{a\in\mathrm{Mat}(2,\CC)\,|\, \mathrm{tr}(a)=0\}=\mathrm{span}_\CC\{\tau_i \, | \, i=1,2,3\},
\end{equation}
with $\{\tau_i\}$ the generators of $\sltwoC$ that are also generators of  $\sutwo$.

Since $\la$ is isomorphic to $\sltwoC$ (as a real algebra), so is $\laCpm$. Let the isomorphism be given by 
\begin{equation}
\label{eq:iso}
	{}^{\pm}\Pj{}{}: \laCpm \rightarrow \sltwoC,
\end{equation}
where we have used the same symbol as for the (anti-)sefdual projector. The two are in fact closely related as will become clear momentarily.\footnote{Furthermore, in the following it will be clear from the index structure, weather the projector or the isomorphism is used.}

Under this isomorphism $a^\pm\in\laCpm$ is mapped to 
\begin{equation}
    \Pjpm{}{}(a^\pm):=\Pjpm{i}{IJ}a^{\pm IJ}\tau_i \in\sltwoC,
\end{equation}
which defines its components $\Pjpm{i}{IJ}$ and $a^{\pm i}:=\Pjpm{i}{IJ}a^{\pm IJ}$. This is invertible, so conversely we can map $x\in\sltwoC$ back to 
\begin{equation}
    \prescript{\pm}{}{P^{-1}}(x)=\Pjpmi{IJ}{i}x^i L^\pm_{IJ} \in\laCpm.
\end{equation}

Combining these two, we can write 
\begin{equation}
    \begin{aligned}
        &a^{\pm IJ}=(P^{-1}\circ P(a))^{IJ}=\Pjpmi{IJ}{i}\Pjpm{i}{KL}a^{\pm KL},\\
        &x^{i}=(P\circ P^{-1}(x))^i=\Pjpm{i}{IJ}\Pjpmi{IJ}{j}x^{j}.
    \end{aligned}
\end{equation}
From this we can conclude the following relations when contracting the isomorphism components over their $\laCpm$ and $\sltwoC$ indices, respectively. Namely,
\begin{equation}\label{eq:isomorphism_identities}
    \begin{aligned}
        &\Pjpmi{IJ}{i}\Pjpm{i}{KL}=\Ppm{IJ}{KL},\\
        &\Pjpm{i}{IJ}\Pjpmi{IJ}{j}=\delta^i_j.
    \end{aligned}
\end{equation}
For the first relation, we do not get just $\delta^I_{[K}\delta^J_{L]}$ but the projector, as selfduality of the index pairs has to be maintained. The projector however acts as $\mathbb{I}$ on selfdual objects.\footnote{We want to make clear that there is a difference when acting with $\Pjpm{i}{KL}$ on $\laC$, compared to $\laCpm$. Applied to $\laC$, elements are not only mapped to $\sltwoC$ but there is an implied projection to the (anti-)selfdual subalgebra, because of the (anti-)\-selfduality of Lorentz indices in $\Pjpm{i}{KL}$. Because of this projection, $P$ is not invertible and therefore cannot be an isomorphism. It is only an invertible and hence an isomorphism on the (anti)-selfdual subalgebras. Everything that follows takes place on the subalgebras.}

Next, we want to show, that the isomorphism does indeed map the respective Cartan-Killing metric onto each other. 
For $x,y\in\sltwoC$ the commutation relation is given by
\begin{equation}
    \commutator{x}{y}=x^iy^j\sce{_{ij}^k}\tau_k. 
\end{equation}
Acting with the isomorphism on the commutator of $a^\pm,b^\pm\in\laCpm$ and employing \eref{eq:isomorphism_identities} yields
\begin{equation}
    \begin{aligned}
    \Ppm{}{}(\commutator{a^\pm}{b^\pm})&=\Pjpm{k}{MN}\commutator{a^\pm}{b^\pm}^{MN}\tau_k\\
    &=a^{\pm IJ}b^{\pm KL}\scfpm{_{IJKL}^{MN}}\Pjpm{k}{MN}\tau_k\\
    &=a^{\pm i}b^{\pm j}\Pjpmi{IJ}{i}\Pjpmi{KL}{j}\scfpm{_{IJKL}^{MN}}\Pjpm{k}{MN}\tau_k.
   \end{aligned} 
\end{equation}
From this we conclude the expected relation of structure constants
\begin{equation}\label{eq:isomorphism_structure_constants}
    \sce{_{ij}^k}=\Pjpmi{IJ}{i}\Pjpmi{KL}{j}\scfpm{_{IJKL}^{MN}}\Pjpm{k}{MN}.
\end{equation}
Using once again \eref{eq:isomorphism_identities}, we obtain the desired relation between the Cartan-Killing metrics,
\begin{equation}\label{eq:isomorphism_CKmetrics}
\begin{aligned}
    k_{ij}&=\sce{_{im}^n}\sce{_{jn}^m}=\Pjpmi{II'}{i}\Pjpmi{JJ'}{j}\scfpm{_{II'KL}^{MN}}\scfpm{_{JJ'MN}^{KL}}\\
    &=\Pjpmi{II'}{i}\Pjpmi{JJ'}{j}k^\pm_{II'JJ'},
\end{aligned}    
\end{equation}
as expected. 

Solving this for one of the inverse isomorphism components, we get
\begin{equation}
\label{eq:isomorphism_CKmetrics_b}
    \Pjpmi{IJ}{i}=\Pjpm{j}{KL} k_{ij} k^{\pm IJKL}.
\end{equation}
We leave this relation as it is, because it is valid without reference to specific choices of bases. Further we do not want to move indices with the Cartan-Killing metrics, but with internal ones, i.\,e. the Minkowski metric and a yet to be specified internal three metric, acting at $\sltwoC$ indices.

Throughout the rest of this article, however, the generators of $\sltwoC$ are given by \begin{equation}
    \tau_j=\frac{\sigma_j}{2i}, 
\end{equation}
with the Pauli matrices $\sigma_j$. Hence the structure constants are just the Levi-Civita symbols $\lc{_{ijk}}$ where one index is raised with the Euclidean metric\footnote{We are this specific in stating the actual form of the structure constants, since we do not want to imply a specific choice of internal metric, yet. The structure constants $\sce{_{ij}^k}$ have a natural positioning of indices, which is given by their definition in terms of the commutation relation. Hence there is no metric involved. Only if we want to express them in terms of e.\,g. the Levi-Civita symbol, we need a notion of raising indices. It is also possible to extract the explicit values of the structure constants from the commutator, compute the Cartan-Killing metric and state the structure constants in terms of the Levi-Civita symbol with an index raised using the Cartan-Killing metric. In the chosen basis we would find $\sce{_{ij}^k}=-2\lc{_{ijm}}k^{mk}$}:
\begin{equation}
    \sce{_{ij}^k}=\lc{_{ijm}}\delta^{mk}.
\end{equation}
Consequently, the explicit form of the Cartan-Killing metric is 
\begin{equation}
    k_{ij}=\lc{_{ikm}}\delta^{ml}\lc{_{jln}}\delta^{nk}=-2\delta_{ij}
\end{equation}
and is also defined with respect to the Euclidean metric. With this prefactor matching the one of $\eta_{I[K}\eta_{L]J}$ in $k_{IJKL}$, \eref{eq:isomorphism_CKmetrics_b} becomes
\begin{equation}\label{eq:isomorphism_CKmetrics_explicit}
    \Pjpmi{IJ}{i}=\Pjpm{j}{KL} \delta_{ij} \eta^{IK}\eta^{LJ}, 
\end{equation}
where we absorbed (anti-)selfduality completely into the contraction with the isomorphism components. As this now involves the metric of the internal Minkowski space, associated with $\laCpm$, it is already at this point suggestive that the metric of the internal space, associated with $\sltwoC$, is the Euclidean metric. We will come back to this later.

Given our  choice of basis, we can determine explicit expressions for the isomorphism components from \eref{eq:isomorphism_identities}, up to a sign $s=\pm 1$, that is unrelated to (anti-)selfduality\footnote{We do not need to include \eref{eq:isomorphism_CKmetrics} here, as this is equivalent to the first relation in \eref{eq:isomorphism_identities}}. The relation between structure constants \eref{eq:isomorphism_structure_constants} uniquely fixes this sign and we end up with
\begin{equation}\label{eq:explicit_isomorphism}
\begin{aligned}
    \Pjpm{j}{KL}&=\begin{cases}
        \mp\frac{i}{2}\delta^j_L & K=0, L=1,2,3 \\
          \frac{1}{2}\lc{^j_{KL}} & K,L=1,2,3
    \end{cases} ,\\
    \Pjpmi{KL}{j}&=\begin{cases}
        \pm\frac{i}{2}\delta_j^L & K=0, L=1,2,3 \\
          \frac{1}{2}\lc{_j^{KL}} & K,L=1,2,3
    \end{cases}.
\end{aligned}  
\end{equation}
Index positions have to be understood with respect to the Euclidean metric. All other components, except the ones obtained from antisymmetry in the Lorentz indices, are vanishing.

\subsubsection{The complex Lorentz group and selfduality}
\label{sec:complex_lorentz_group}
With the isomorphism of Lie algebras \eref{eq:iso} are able to map an (anti-)selfdual Lorentz connection to an $\sltwoC$-valued field. However, we have to make sure that this object is indeed a connection as well. In order to see this, we need to have the correct behaviour under gauge transformation with respect to $\SLtwoC$. 

At first we have to specify the transformation behaviour of the complex Lorentz connection $\omega$, which we started with. 
Recall that the complex Lorentz group 
\begin{equation}
    \LGC=\curlybrackets{M\in\mathrm{Mat}(4,\CC)\, |\, M^T\eta M=\eta, \det(M)=1}
\end{equation}
is connected \cite{jost_pauli-prinzip_1960}. Therefore every element can be expressed by a product of exponentials of its Lie algebra, which is the complexified Lorentz Lie algebra $\laC$. 
Given $g\in\LGC$, we can decompose the corresponding Lie algebra elements $a_i=a^+_i+a^-_i$ into its commuting selfdual and anti-selfdual parts. This yields
\begin{equation}
\label{eq:LG_selfdual_decomp}
    g=\prod_i e^{a_i}=\prod_i e^{a^+_i+a^-_i}=\prod_i e^{a^+_i}\prod_j e^{a^-_j}=:g^+g^-.
\end{equation}
The group elements $g^\pm=\prod_i\exp{a^\pm_i}$ are commuting. This motivates to define
\begin{equation}
\label{eq:Gpm}
    G_\pm:=\curlybrackets{g_\pm\in\LGC\bigg| g_\pm=\prod_i e^{a^\pm_i}, a^\pm_i\in \laCpm }.
\end{equation}
So $G_\pm$ are all the elements of $\LGC$ which can be associated to selfdual or anti-selfdual Lie algebra elements, respectively. They inherit multiplication, inverse elements and the neutral element from $\LGC$. Therefore they are commuting subgroups of $\LGC$. By construction they are connected.  

We want to clarify how $G_\pm$ are related to $\SLtwoC$. Since they arise as exponentials of a Lie algebra isomorphic to $\SLtwoC$, they are the operators of a 4d representation of $\SLtwoC$. The Casimir elements are 
\begin{equation}
    \begin{aligned}
        C^\pm_1 &= k^{IJKL}L^\pm_{IJ}L^\pm_{KL}, \\
        C^\pm_2 &= k^{IJKL}L^\pm_{IJ}(\ast L^\pm)_{KL}=\pm iC^\pm_{1},
    \end{aligned}
\end{equation}
where we used the (anti-)selfduality of the generators. Using the fact that the (anti-)selfdual generators are given by the projectors, see \eref{eq:selfdual_generators_laCpm}, and the relation between Cartan-Killing and Minkowski metric, it is easy to determine 
\begin{equation}
    k^{IJKL}\T{(L^\pm_{IJ})}{^M_{M'}}\T{(L^\pm_{KL})}{^{M'}_N}=-\frac{1}{2}\Pmp{MM'}{NM'}=\frac{1}{4}\rb{\delta^{[M}_{N}\delta^{M']}_{M'}-\frac{i}{2}\lc{^{MM'}_{NM'}}}=\frac{3}{8}\delta^M_N.
\end{equation}
In the last step the trace of the Levi-Civita symbol vanishes. With this, the two Casimir elements are given by 
\begin{equation}\label{eq:Casimirs}
        \begin{aligned}
            \T{{C^\pm_1}}{^M_N} &=\frac{3}{8}\delta^M_N, \\
        \T{{C_2^\pm}}{^M_N} &= \pm i\frac{3}{8}\delta^M_N.
    \end{aligned}
\end{equation}
Next we introduce 
\begin{equation}
    \T{{(K^\pm_i)}}{^M_N}:=\T{{(L^\pm_{0i})}}{^M_N}=\T{\Ppm{}{}}{_{0i}^M_N}
\end{equation}
and
\begin{equation}
        \T{{(J^\pm_i)}}{^M_N}:=\T{(\ast L^\pm_{0i})}{^M_N}=\frac{1}{2}\lc{_{0i}^{jk}}\T{{(L^\pm_{jk})}}{^M_N}=\pm i\T{{(L^\pm_{0i})}}{^M_N} =\pm i\T{{(K^\pm_i)}}{^M_N}.
\end{equation}
Using this decomposition of the (anti-)selfdual basis, it is straightforward to see that we can express the Casimir elements in the additional forms
\begin{equation}
    \begin{aligned}
        C^\pm_1 &=K_\pm^2-J_\pm^2=2K_\pm^2=-2J_\pm^2, \\
        C_2^\pm &= 2K_\pm\cdot J_\pm=\pm 2iK_\pm^2=\mp 2iJ_\pm^2.
    \end{aligned}
\end{equation}
Here all products are with respect to the Euclidean three-metric. In comparison to the form of standard Casimirs of the Lorentz algebra, this suggests that $K^\pm$ and $J^\pm$ have indeed an interpretation as generators of boosts and rotations, respectively, even if they have been obtained from the (anti-)selfdual generators.\footnote{The reason for this is \eref{eq:P_is_a_homomorphism}.} 

This can however be verified when determining the commutation relations for $K^\pm$ and $J^\pm$. Using the commutation relation of (anti-)selfdual generators and the relation between $K^\pm$ and $J^\pm$, we find
\begin{equation}
    \begin{aligned}
        &\commutator{K^\pm_i}{K^\pm_j}=\pm \frac{i}{2} \lc{_{ij}^k} K^\pm_k=\frac{1}{2}\lc{_{ij}^k}J^\pm_k,\\
        &\commutator{J^\pm_i}{J^\pm_j}=(\pm i)^2 \commutator{K^\pm_i}{K^\pm_j}=-\frac{1}{2}\lc{_{ij}^k}J^\pm_k,\\
        &\commutator{K^\pm_i}{J^\pm_j}=\pm i \commutator{K^\pm_i}{K^\pm_j}=\pm \frac{i}{2} \lc{_{ij}^k}J^\pm_k=-\frac{1}{2} \lc{_{ij}^k}K^\pm_k.
    \end{aligned}
\end{equation}
Hence $K_\pm$ and $J_\pm$ indeed satisfy the commutation relations for boosts and rotations. In this form, the relations are the same for both,  self-dual and anti-selfdual algebra.

In the usual way we then recombine these into generators $M^\pm_j$ and $N^\pm_j$   of two commuting copies of $\sltwoC$. It turns out that $M^+_j=J^+_j$ and $N^-_j=J^-_j$, while the other generators vanish. By relating $J^{\pm}_j$ to the actual spin operators, we can determine the representation label of the Casimir elements in \eref{eq:Casimirs} to be $j=\frac{1}{2}$.  Therefore we have a multiplicity of two for this representation, in order to construct the four dimensional representation. 
We therefore find that $\laCp$ corresponds to the 
\begin{equation}
    \pi_+=\pi_{(\frac{1}{2},0)}\oplus\pi_{(\frac{1}{2},0)}
\end{equation}
representation of $\sltwoC$, while $\laCm$ corresponds to the
\begin{equation}
    \pi_-=\pi_{(0,\frac{1}{2})}\oplus\pi_{(0,\frac{1}{2})}
\end{equation}
representation.
Consequently there is indeed a basis in which one of the corresponding group representations is of diagonal form. Let us say this is the case for $G_+$. Let $g^+=\mathbb{I}_2\otimes g$ and $g^-=\Bar{g}\otimes\mathbb{I}_2$.
With this, we understand how $G_\pm$ are four dimensional representations of $\SLtwoC$. Hence the decomposition \eref{eq:LG_selfdual_decomp}, given a specific basis, can be understood as 
\begin{equation}
    g=g^+ g^-= \Pi_+(g_1)\Pi_-(g_1)=(\mathbb{I}_2\otimes g_1)(\Bar{g}_2\otimes\mathbb{I}_2)=\Bar{g}_2\otimes g_1.
\end{equation}
In this tensor product form we immediately see that
\begin{equation}
    \Bar{g}_2\otimes g_1=-\Bar{g}_2\otimes (-g_1).
\end{equation}
Hence the decomposition of $\LGC$ into $G_\pm$ is not unique, but takes the form 
\begin{equation}
    \LGC=G_+\times G_-/H\cong\SLtwoC\times\SLtwoC/H,
\end{equation}
with $H=\{\mathbb{I}_4,-\mathbb{I}_4\}$. This is a well known result, see e.\,g. \cite{sexl_relativity_2001}, which we arrived at, using the (anti-)selfdual decomposition of the complexified Lorentz Lie algebra.

\subsubsection{From Lorentz to $\SLtwoC$-connections}
\label{sec:From_Lorentz_to_SLtwoC_connections}
In section \ref{sec:Relation_sltwoC} we saw how to express $\laCpm$ objects in terms of $\sltwoC$. In this section we want to establish how the selfdual and anti-selfdual Lorentz connections are related to $\SLtwoC$ connections.

The last section tells us that there is a basis in which we can represent the $\laC$-connection $\omega$ in the form 
\begin{equation}
    \omega=\II\otimes A_++\Bar{A}_-\otimes\II,
\end{equation}
which ensures the commutativity of the selfdual and anti-selfdual parts, represented by the $\sltwoC$ elements $A_+$ and $A_-$, respectively. 
Understanding again $g$ as $\Bar{g}_-\otimes g_+$, the transformation behaviour of $\omega$ takes the form 
\begin{equation}
\begin{aligned}
    \omega'&=g\omega g^{-1} + g\grad g^{-1}=\II\otimes(g_+A_+g_+^{-1}+g_+\mathop{d}g_+^{-1})+
    (\Bar{g}_-\Bar{A}_-\Bar{g}_-^{-1}+\Bar{g}_-\mathop{d}\Bar{g}_-^{-1})\otimes\II\\
    &=\II\otimes A'_++\Bar{A}'_-\otimes\II.
\end{aligned}    
\end{equation}
Hence the (anti-)selfdual parts of the connection $\omega$ are not only connections on their own, but are indeed already $\sltwoC$-connections.

We note that our arguments so far only establish the existence of the $\sltwoC$-connections $\omega_\pm$ locally. To give a rigorous argument, one would have to establish the existence of the two $\SLtwoC$-bundles and their connections globally. We see no obstruction to doing this, and might come back to it elsewhere. Remarkably, it seems that no additional global structure, such as a spin structure, is needed when starting from the theory with $\LGC$ as structure group. This is in contrast to the real Ashtekar-Barbero formulation \cite{barbero_g_real_1995}, where a spin structure is needed to obtain the $\SUtwo$-formulation \cite{fleischhack_ashtekar_2011}.

\subsection{Selfdual Palatini action}
\label{sec:SelfdualPalatiniAction}

With the projectors that were introduced in order to split the complexified Lorentz algebra into its selfdual and anti-selfdual components, we can split the action of complex general relativity, respectively. 

\subsubsection{Selfdual split and $\LGCp$ variables}

Using that the projectors sum up to the identity in \eref{eq:projection_properties}, 
the complex Palatini action decomposes into a selfdual and an anti-selfdual part:

\begin{equation}
\begin{aligned}
    S_\CC[\omega,e]&=\frac{\lambda}{\kappa}\int_\curlyM \epsilon_{IJKL}\Sigma^{IJ}\!\wedge F^{KL}\\
    &=\frac{\lambda}{\kappa}\int_\curlyM \epsilon_{IJKL}\rb{\Pp{IJ}{I'J'}+\Pm{IJ}{I'J'}}\rb{\Pp{KL}{K'L'}+\Pm{KL}{K'L'}}\Sigma^{I'J'}\!\wedge F^{K'L'}\\
    &=\begin{aligned}[t]&\frac{\lambda}{\kappa}\int_\curlyM \epsilon_{IJKL}\Pp{IJ}{I'J'}\Pp{KL}{K'L'}\Sigma^{I'J'}\!\wedge F^{K'L'}\\
    &+\frac{\lambda}{\kappa}\int_\curlyM \epsilon_{IJKL}\Pm{IJ}{I'J'}\Pm{KL}{K'L'}\Sigma^{I'J'}\!\wedge F^{K'L'}.\end{aligned}
\end{aligned}
\end{equation}
Introducing the projected Plebanski two-form $\Sigma^{\pm IJ}$ and the (anti-)selfdual projection of the curvature curvature $F^{\pm IJ}$, the complex action decomposes into
\begin{equation}
	S=S_++S_-.
\end{equation}
The (anti-)selfdual actions $S_\pm$ are 
\begin{equation}\label{eq:Palatini_pm}
    S_\pm[\omega^{(\pm)},e]=\frac{\lambda}{\kappa}\int_\curlyM \epsilon_{IJKL}\Sigma^{(\pm)IJ}\!\wedge F^{(\pm)KL}.
\end{equation}
While the projection $\Sigma^{(\pm) IJ}=\Ppm{IJ}{KL}\Sigma^{KL}$ is straightforward, there is some more work required, in order to show that $F^{(\pm)IJ}$ is indeed the curvature of the (anti-)selfdual connection $\omega^{(\pm)IJ}=\Ppm{IJ}{KL}\omega^{KL}$.
For the second term in \eref{eq:curvature_LorentzConnection} we find that
\begin{equation}
    \T{\omega}{^I_M}\wedge\T{\omega}{^{MJ}}=\delta^{[I}_{[K}\eta_{L][M}\delta^{J]}_{N]}\T{\omega}{^{KL}}\wedge\T{\omega}{^{MN}}=\frac{1}{2}\scf{_{KLMN}^{IJ}}\T{\omega}{^{KL}}\wedge\T{\omega}{^{MN}}.
\end{equation}
Thus
\begin{equation}
    \Ppm{IJ}{I'J'} \T{\omega}{^{I'}_M}\wedge\T{\omega}{^{MJ'}}= \T{\omega}{^{(\pm)I}_M}\wedge\T{\omega}{^{(\pm)MJ}},
\end{equation}
where we used \eref{eq:structure_constants_laCpm} in order to pull the projectors through the structure constant, which again in absorbed in the contraction of indices. So indeed from \eref{eq:curvature_LorentzConnection} we get
\begin{equation}
    \Ppm{IJ}{KL}F[\omega]^{KL}=F[\omega^{(\pm)}]^{IJ}.
\end{equation}

We want to make a remark about the use of internal metrics here.
Given the considerations about the Cartan-Killing metric of $\laCpm$, one can ask oneself, whether is possible to write down the Palatini action \eref{eq:Palatini_pm} referring to the Cartan-Killing metric only. This metric, however, is only able to move antisymmetric pairs of indices and compared to eq. \eref{eq:curvature_LorentzConnection} we need to raise a single index of the curvature of the Lorentz connection. For this, however, we need the internal Minkowski metric.

In the following we will restrict ourselves to the selfdual part of the action only, because it yields a theory that reduces to Einstein gravity under certain reality conditions. Out of convenience, we will also drop the $+$ indicating selfdual objects. 

Next we look at the projection of the Plebanski two-form $\Sigma^{IJ}$, i.\,e. the selfdual part of $e^I\wedge e^J$. It is given by
\begin{equation}
    \Sigma^{IJ}=\Pud{IJ}{KL}e^K_\alpha e^L_\beta \mathop{dx^\alpha}\!\wedge\mathop{dy^\beta}=\frac{1}{2}\rb{e^{[I}_\alpha e^{J]}_\beta-\frac{i}{2}{\epsilon^{IJ}}_{KL}e^K_\alpha e^L_\beta}\mathop{dx^\alpha}\!\wedge\mathop{dy^\beta}
\end{equation}
From this we define the following\footnote{Note, that, technically, the $\Sigma_{ab}^{IJ}$ defined in this way differ from the components of the Plebanski form by a (convenient) factor of 1/2.}
\begin{equation}
    \Sigma^{IJ}_{\alpha\beta}:=\frac{1}{2}\rb{e^{[I}_\alpha e^{J]}_\beta-\frac{i}{2}{\epsilon^{IJ}}_{KL}e^K_\alpha e^L_\beta}.
\end{equation}
In order to further fix the notation, $\Sigma^{\alpha\beta}_{IJ}$, where indices where moved with the respective metrics, will be referred to as Plebanski bi-vector. 

Since the frame fields are soldering forms between the internal and the tangential space, it is natural to ask how selfduality of internal indices transfers to tangential ones. To this end, we first express the determinant of the spacetime metric in terms of the frame fields. This yields the usual relation
\begin{equation}
	\det\rb{g_{\alpha\beta}}=-\det\rb{e^I_\alpha}^2
\end{equation}
even though $g_{\alpha\beta}$ and $e^I_\alpha$ are complex quantities now. Taking the square root we can write this as 
\begin{equation}\label{eq:det_e_g_relation}
	\sqrt{-\det\rb{g_{\alpha\beta}}}=s_e\det\rb{e^I_\alpha},
\end{equation} 
where we introduced a sign function for the determinant of complex tetrads, i.\,e.
\begin{equation}\label{eq:sign_det_e}
	s_e:=\frac{\det\rb{e^I_a}}{\sqrt{\det\rb{e^I_a}^2}}.
\end{equation}
It is real and takes values $\pm 1$. Equation \ref{eq:det_e_g_relation} further shows that $\det\rb{e^I_a}$ carries a density weight of $+1$.

In order to relate the Hodge dual with respect to internal and tangential indices, respectively, we need the following identity for the Levi-Civita symbol:
\begin{equation}\label{eq:LC_general_delta_identity}
    \lcdp{^{\alpha'\beta'\gamma'\delta'}}\lcdm{_{\alpha \beta \gamma \delta}}=-4! \delta^{\alpha'}_{[\alpha}\delta^{\beta'}_\beta\delta^{\gamma'}_\gamma\delta^{\delta'}_{\delta]}.
\end{equation}
Here the tilde indicates the density weight of $\pm 1$ and the sign on the right hand side comes from all primed indices being raised with the Minkowski metric. Further, as the density weights cancel, this is also valid for the Levi-Civita tensor and with just internal indices.  Therefore
\begin{equation}
\begin{aligned}
   &\lc{_{IJKL}}e^I_\alpha e^J_\beta e^K_\gamma e^L_\delta= \lc{_{IJKL}}e^I_{\alpha'} e^J_{\beta'} e^K_{\gamma'} e^L_{\delta'}
   \delta^{\alpha'}_{[\alpha}\delta^{\beta'}_\beta\delta^{\gamma'}_\gamma\delta^{\delta'}_{\delta]}\\
   &= -\frac{1}{4!} \lc{_{IJKL}} \lcdp{^{\alpha'\beta'\gamma'\delta'}} e^I_{\alpha'} e^J_{\beta'} e^K_{\gamma'} e^L_{\delta'} \lcdm{_{\alpha \beta \gamma \delta}}\\
   &=-\det\rb{e^I_\alpha}\lcdm{_{\alpha \beta \gamma \delta}}.
\end{aligned}   
\end{equation}
So using \eref{eq:det_e_g_relation} and cancelling the density weight of the Levi-Civita symbol, we find 
\begin{equation}
    \lc{_{IJKL}}e^I_\alpha e^J_\beta e^K_\gamma e^L_\delta=-s_e\lc{_{\alpha \beta \gamma \delta}}.
\end{equation}
  
From this we can determine the action of the tangential Hodge star:
\begin{equation}
\begin{aligned}
    &\rb{\ast_{\mathrm{ST}}\Sigma}^{IJ}_{\alpha\beta}=\frac{1}{2}g_{\alpha\alpha'}g_{\beta\beta'}\lc{^{\alpha'\beta'\gamma\delta}}\frac{1}{2}\rb{e^{[I}_\gamma e^{J]}_\delta-\frac{i}{2}{\epsilon^{IJ}}_{KL}e^K_\gamma e^L_\delta}\\
    &=\frac{-s_e}{4}g_{\alpha\alpha'}g_{\beta\beta'}\rb{-\lc{^{IJKL}} e^{\alpha'}_K e^{\beta'}_L - \frac{i}{2} \lc{^{IJ}_{KL}} \lc{^{MNKL}} e^{\alpha'}_M e^{\beta'}_N   }\\
    &=\frac{-s_e}{4}g_{\alpha\alpha'}g_{\beta\beta'}\rb{-\lc{^{IJKL}} e^{\alpha'}_K e^{\beta'}_L + 2i \eta^{II'}\eta^{JJ'} e^{\alpha'}_{[I'}e^{\beta'}_{J']}}  \\
    &=-is_e \frac{1}{2}\rb{e^{[I}_\gamma e^{J]}_\delta-\frac{i}{2}{\epsilon^{IJ}}_{KL}e^K_\gamma e^L_\delta}
\end{aligned}    
\end{equation}
Note that we used the usual identities for the Levi-Civita symbol from the second to the third line. Therefore we have
\begin{equation}\label{eq:spacetime_antiselfduality}
\rb{\ast_{\mathrm{ST}}\Sigma}^{IJ}_{\alpha\beta}=-is_e \Sigma^{IJ}_{\alpha\beta}
\end{equation}
and hence internal selfduality is only carried over up to a minus sign and depends on the orientation of the complex frame fields. 

Before expressing the selfdual action in terms of $\sltwoC$ variables, we rewrite it as follows. As a first step we use the selfduality of the Plebanski two-form:
\begin{equation}
    S=\frac{\lambda}{\kappa}\int_\curlyM \epsilon_{IJKL}\Sigma^{IJ}\!\wedge F^{KL}
    =\frac{2i\lambda}{\kappa}\int_\curlyM \Sigma_{IJ}\!\wedge F^{IJ}.
\end{equation}
As a second step, we use \eref{eq:LC_general_delta_identity} to separate the four-form $\grad{^4x}=\frac{1}{4!}\lcdm{_{\mu\nu\rho\sigma}}\grad{x^\mu}\!\wedge\grad{x^\nu}\!\wedge\grad{x^\rho}\!\wedge\grad{x^\sigma}$ from the components: 
\begin{equation}
    S=-\frac{2i\lambda}{\kappa}\int_\curlyM\!\grad{^4x} \lcdp{^{\alpha\beta\gamma\delta}} \Sigma_{IJ\alpha\beta}\frac{1}{2}F^{IJ}_{\gamma\delta}.
\end{equation}
Further we employ \eref{eq:spacetime_antiselfduality} in order to get rid of the spacetime Levi-Civita symbol. This yields the final form of the selfdual action:
\begin{equation}\label{eq:selfdual_action_no_wedge}
    S=-\frac{2\lambda}{\kappa}\int_\curlyM \dxd{4}{x}\sqrt{-g}\, s_e\Sigma_{IJ}^{\alpha\beta} F^{IJ}_{\alpha\beta}.
\end{equation}
The determinant factor $\sqrt{-g}$ arises due to the transition from the Levi-Civita density to its tensor counterpart, needed for the spacetime Hodge-$\ast$.

The next step is to go over from $\laCp$- to $\sltwoC$-variables.

\subsubsection{$\SLtwoC$ variables and $3+1$ split}

In the contraction of internal indices in \eref{eq:selfdual_action_no_wedge}, there is an implicit selfdual projection. We can now use the first relation in \eref{eq:isomorphism_identities} in order to rewrite this contraction in terms of $\sltwoC$-indices:
\begin{equation}
    \Sigma_{IJ}^{\alpha\beta}\T{P}{^{IJ}_{KL}} F^{KL}_{\alpha\beta}=\Sigma_{IJ}^{\alpha\beta} \Pji{IJ}{i}\Pj{i}{KL} F^{KL}_{\alpha\beta}.
\end{equation}
This lets us define the $\sltwoC$-valued Plebanski bi-vector 
\begin{equation}
    \Sigma_i^{\alpha\beta}:=\Pji{IJ}{i}\Sigma_{IJ}^{\alpha\beta}.
\end{equation}
Further we have the action of the isomorphism on the curvature components, that yields the curvature of the $\SLtwoC$ connection 
\begin{equation}
    A^i_\alpha:=\Pj{i}{IJ}\omega^{IJ}_\alpha,
\end{equation}
which was introduced in section \ref{sec:From_Lorentz_to_SLtwoC_connections}.
We use the fact that the curvature components can be expressed in terms of the commutator of the $\laCp$ connection and hence find
\begin{equation}
\begin{aligned}
    \Pj{i}{KL} F^{KL}_{\alpha\beta}&=\Pj{i}{KL}2\partial_{[\alpha}\omega^{KL}_{\beta]}+\Pj{i}{KL}\commutator{\omega_\alpha}{\omega_\beta}^{KL}\\
    &=2\partial_{[\alpha}A_{\beta]}^i+\commutator{A_\alpha}{A_{\beta}}^i\\
    &=2\partial_{[\alpha}A_{\beta]}^i+\sce{_{jk}^i}A^j_\alpha A^k_\beta=: F^{i}_{\alpha\beta}[A].
\end{aligned}    
\end{equation}
This allows to recast the selfdual action, based on $\laCp$-variables, in terms of $\sltwoC$-valued $A^i_\alpha$ and $\Sigma^{\alpha\beta}_i$.

The next step is to perform a $3+1$-split of this action.
In the usual way, see e.\,g. \cite{ashtekar_lectures_1998,ashtekar_background_2004,thiemann_modern_2007}, we introduce a foliation
\begin{equation}
    t^\alpha=N n^\alpha + N^\alpha,\quad N^\alpha n_\alpha =0
\end{equation}
with respect to lapse function $N$, shift vector $N^a$ and the normal vector of spatial slices $n^a$. We want to emphasise that this foliation is real, i.\, e. lapse, shift and normal are real objects.  The spacetime manifold therefore splits into $\curlyM=M\times\RR$, where $M$ is the spatial manifold.
This gives rise to the spatial metric
\begin{equation}
    q_{\alpha\beta}=g_{\alpha\beta}+n_\alpha n_\beta,\quad q_{\alpha\beta}n^\beta=0.
\end{equation}
Projections to the spatial manifold are done by 
\begin{equation}
     q^\alpha_\beta=\delta^\alpha_\beta+n^\alpha n_\beta
\end{equation}
and hence a $3+1$ split can be performed by using the relation
\begin{equation}\label{eq:delta_spatial_projection}
     \delta^\alpha_\beta=q^\alpha_\beta-n^\alpha n_\beta.
\end{equation}
Regarding the metric determinant we have the following identity, which is still valid for the complex spacetime and spatial metrics:
\begin{equation}
    \sqrt{-g}=N\sqrt{q}.
\end{equation}
Inserting \eref{eq:delta_spatial_projection} in the contraction of spacetime indices in the action therefore yields
\begin{equation}
\begin{aligned}
    S&=-\frac{2\lambda}{\kappa}\int_\curlyM \dxd{4}{x} \sqrt{q} N s_e  \Sigma_{i}^{\alpha\beta} F^{i}_{\gamma\delta}(q_\alpha^\gamma-n_\alpha n^\gamma)(q_\beta^\delta-n_\beta n^\delta)\\
    &=-\frac{2\lambda}{\kappa}\int_\curlyM \dxd{4}{x} \sqrt{q} N s_e \Sigma_{i}^{\alpha\beta} F^{i}_{\gamma\delta} (q_\alpha^\gamma q_\beta^\delta-q_\alpha^\gamma n_\beta n^\delta-q_\beta^\delta n_\alpha n^\gamma + n_\alpha n_\beta n^\gamma n^\delta)\\
    &=-\frac{2\lambda}{\kappa}\int_\curlyM \dxd{4}{x} \sqrt{q} N s_e  \Brb{q_\alpha^\gamma q_\beta^\delta \Sigma_{i}^{\alpha\beta} F^{i}_{\gamma\delta} -2 q_\alpha^\gamma n_\beta n^\delta \Sigma_{i}^{\alpha\beta} F^{i}_{\gamma\delta}} \\
    &=-\frac{2\lambda}{\kappa}\int_\curlyM \dxd{4}{x}  \sqrt{q} s_e \Brb{ N q_\alpha^\gamma q_\beta^\delta \Sigma_{i}^{\alpha\beta} F^{i}_{\gamma\delta} + 2   q_\alpha^\gamma n_\beta N^\delta \Sigma_{i}^{\alpha\beta} F^{i}_{\gamma\delta} -2 t^\delta  q_\alpha^\gamma n_\beta \Sigma_{i}^{\alpha\beta} F^{i}_{\gamma\delta}}
\end{aligned}    
\end{equation}
From the second to the third line, we used the antisymmetry of $\Sigma^{i}_{\alpha\beta} F_{i}^{\gamma\delta}$ to combine the middle two terms in the bracket and to cancel the symmetric combination of $n$'s. From the third to the fourth line we used $Nn^\beta=t^\beta-N^\beta$. Using the easy to show relation $t^\alpha F^i_{\alpha\beta}= \mathop{\curlyL_t}(A^i_\beta)-\mathop{\curlyD_\beta} (t^\alpha A^i_\alpha)$, where $\curlyD$ is the covariant derivative with respect to $A^i_a$, and taking the antisymmetry of $F_{\alpha\beta}$ into account, we arrive at 
\begin{align}
    S=\int_\curlyM \dxd{4}{x}  \Bigl( &-\frac{4\lambda\sqrt{q} s_e}{\kappa} q_\alpha^\gamma n_\beta \Sigma^{\alpha\beta}_i \mathop{\curlyL_t}(A^i_\gamma)\label{eq:canonical_pair_piece}\\
    &-  \frac{2\lambda\sqrt{q} s_e}{\kappa} N  q_\alpha^\gamma q_\beta^\delta \Sigma_{i}^{\alpha\beta} F^{i}_{\gamma
\delta}\label{eq:hamiltonian_piece}\\
    &-  \frac{4\lambda\sqrt{q} s_e}{\kappa} N^\delta  q_\alpha^\gamma n_\beta \Sigma_{i}^{\alpha\beta} F^{i}_{\gamma\delta}\label{eq:diffeo_piece}\\
    &+ \frac{4\lambda\sqrt{q} s_e}{\kappa}  q_\alpha^\gamma n_\beta \Sigma_{i}^{\alpha\beta} \mathop{D_\gamma} (t^\delta A^i_\delta)\Bigr)\label{eq:gauss_piece}.
\end{align}
From the first we will extract the canonically conjugated variables and the symplectic structures. The terms proportional to lapse and shift will be identified as Hamilton and diffeomorphism constraints, respectively, and the last term will form the Gau{\ss} constraint. We will look at the individual contributions separately in the next sections.

\subsubsection{Canonical pair and Poisson relation}
We start with the piece \eref{eq:canonical_pair_piece}, from which we will get the symplectic structure and the actual canonical pair consisting of the Ashtekar connection $A^i_a$ and a corresponding electric field $E^a_i$.

Using $\mathop{\curlyL_t}(q^\alpha_\beta)=0$ and the spacetime selfduality of the Plebanski two-form we find that
\begin{equation}
\begin{aligned}
    &-\int_\curlyM \dxd{4}{x} \frac{4\lambda\sqrt{q} s_e}{\kappa} q_\alpha^\gamma n_\beta \Sigma^{\alpha\beta}_i \mathop{\curlyL_t}(A^i_\gamma) = \\
    & =-\int_\curlyM \dxd{4}{x} \frac{4\lambda\sqrt{q} s_e}{\kappa} \frac{i}{2 s_e}  n_\beta \lc{^{\alpha\beta}_{\gamma\delta}} \Sigma^{\gamma\delta}_i \mathop{\curlyL_t}(q_\alpha^\rho A^i_\rho).
\end{aligned}    
\end{equation}
Here we have $q_\alpha^\rho A^i_\rho$, the spatial projection of the spacetime Ashtekar connection. All open indices of $n_\beta \lc{^{\alpha\beta}_{\gamma\delta}}=-n_\beta \lc{^{\beta\alpha}_{\gamma\delta}}$ are implicitly spatially projected and so is therefore $\Sigma^{\gamma\delta}_i$. Hence we pull everything back to the spatial manifold and replace the spacetime Levi-Civita tensor by its spatial analogue $\lc{^{a}_{ab}}$. This yields 
\begin{equation}\label{eq:symplectic_structure_AE}
    \begin{aligned}
        &\int_\RR\grad{t}\int_M \grad{^3x} \frac{2\lambda\sqrt{q} i}{\kappa}   \lc{^{a}_{bc}} \Sigma^{bc}_i \mathop{\curlyL_t}(A^i_a)=\\
        &=\int_\RR\grad{t}\int_M \grad{^3x}\frac{2\lambda}{\kappa i} \tE^a_i\mathop{\curlyL_t}(A^i_a),
    \end{aligned}
\end{equation}
where $A^i_a$ is the pullback of $q_\alpha^\rho A^i_\rho$ to the spatial manifold. It is the spatial $\sltwoC$-valued, hence complex, Ashtekar connection. Recalling from where we started, this is the pullback of the selfdual part of a complex spacetime Lorentz connection, expressed in terms of $\sltwoC$. In comparison, the real Ashtekar-Barbero connection is a purely spatial object and is not the pullback of  a spacetime connection\footnote{~See for example \cite{thiemann_modern_2007}.}.

Further we introduced the electric field $\tE^a_i$, a density of weight one, corresponding to the Ashtekar connection. It is given by
\begin{equation}\label{eq:electric_field}
    \tE^a_i=-\sqrt{q}\lc{^a_{bc}}\Sigma^{bc}_i. 
\end{equation}

Now we can read off the  symplectic structure and invert it for a Poisson relation
\begin{equation}\label{eq:Poisson_Relation_AE}
    \PB{A^i_a(x)}{\tE^b_j(y)}=\frac{i\kappa}{2\lambda}\delta^b_a\delta^i_j\delta(x,y).
\end{equation}
This tells us that the canonically conjugated variables of selfdual general relativity are the $\sltwoC$ connection $A^i_a$ and the corresponding electric field $E^a_i$ which is complex valued as well. 

At this point we see the effect of rescaling the action with $\lambda$, which enters our Poisson relation. There are now two suggestive cases and these are the only ones of interest throughout the rest of this paper. We look at $\lambda=1$ and $\lambda=i$. With this we have two distinct Poisson relations:
\begin{equation}\label{eq:PB_with_i}
    \PB{A^i_a(x)}{\tE^b_j(y)}=\frac{i\kappa}{2}\delta^b_a\delta^i_j\delta(x,y),
\end{equation}
if $\lambda=1$ or 
\begin{equation}\label{eq:PB_without_i}
    \PB{A^i_a(x)}{\tE^b_j(y)}=\frac{\kappa}{2}\delta^b_a\delta^i_j\delta(x,y)
\end{equation}
if $\lambda=i$. The first case \eref{eq:PB_with_i} is the Poisson relation that is expected when using selfdual variables, see for example \cite{ashtekar_lectures_1998,ashtekar_gravitational_2021}. In the second case however, we enforce the Poisson relation as for real Ashtekar-Barbero variables. Nevertheless, connection and electric field in \eref{eq:PB_without_i} are complex. In \cite{thiemann_account_1995,thiemann_reality_1996} this kind of Poisson relation for complex variables appears in the context of using complexifier methods in order to construct a transformation  from real to complex LQG, which already incorporates the reality conditions.  
A further instance of the use of this Poisson bracket for complex variables can be found in \cite{wieland_complex_2012}, but under a different premise.

Before we go on and show that the electric field indeed act as triads of the spatial manifold, we need to analyse this complex Poisson relation in more detail. We do this in the next section.

\subsubsection{Holomorphicity}
Looking at the term \eref{eq:symplectic_structure_AE} giving the symplectic potential, we get a kinematical phase space in which Ashtekar connection and the electric field are canonical coordinates. Both are a priori complex. The question arises whether it is possible to derive from \eref{eq:symplectic_structure_AE} a symplectic structure for real and imaginary parts of the connection and the electric field.

We split $A^i_a={A_R}^i_a+i{A_I}^i_a$ and $\tilde{E}^b_j={\tilde{E}_R}^b_j+i{\tilde{E}_I}^b_j$ into the respective real and imaginary parts and insert this into  \eref{eq:symplectic_structure_AE}:
\begin{equation}
    \tilde{E}^a_i \curlyL_t {A}^i_a= {\tilde{E}_R}^a_i \curlyL_t {A_R}^i_a -{\tilde{E}_I}^a_i \curlyL_t {A_I}^i_a+i{\tilde{E}_I}^a_i \curlyL_t {A_R}^i_a +i{\tilde{E}_R}^a_i \curlyL_t {A_I}^i_a. 
\end{equation}
From this we determine a pre-symplectic structure for the variations $\delta A^i_a$, $\delta\tE^a_i$ split into their respective real an imaginary parts $\delta^\mathrm{R}A^i_a$, $\delta^\mathrm{I}A^i_a$ and  $\delta^\mathrm{R}\tE^a_i$, $\delta^\mathrm{I}\tE^a_i$, respectively. It is given by
\begin{equation}\Omega(\delta_1^\mathrm{R},\delta_1^\mathrm{I},\delta_2^\mathrm{R},\delta_2^\mathrm{I})=\int_\Sigma\grad{^3x}  \delta^\mathrm{T}_1(x)M\delta_2(x),
\end{equation}
where $\delta_i^\mathrm{T}:=(\delta_i^\mathrm{R}A^1_1,\delta_i^\mathrm{I}A^1_1,\delta_i^\mathrm{R}\tE^1_1$, $\delta_i^\mathrm{I}\tE^1_1,\dots)$. The $12\times 12$-Matrix $M$, characterising the presymplectic form is given by
\begin{equation}\label{eq:symplectic_matrix}
    M:=\II_3\otimes\begin{pmatrix} 0 & 0 & 1 & i  \\ 
                        0 & 0 & i & -1 \\
                        -1 & -i & 0 & 0 \\
                        -i & 1 & 0 & 0 \end{pmatrix}.
\end{equation}
The corresponding Poisson relation for phase space functions $F$ and $G$ would therefore be of the form 
\begin{equation}
    \PB{F}{G}=\int_\Sigma\grad{^3x}\, ({\nabla F})^\mathrm{T}(x)\, M^{-1}\,\nabla G(x),
\end{equation}
where $\nabla F$ and $\nabla G$ are the collections of functional derivatives, according to 
\begin{equation}
    \nabla^\mathrm{T}:=\rb{\frac{\delta}{\delta {A^\mathrm{R}}^1_1(x)},\frac{\delta}{\delta {A^\mathrm{I}}^1_1(x)},\frac{\delta}{\delta {\tE^\mathrm{R}}^1_1(x)},\frac{\delta}{\delta {\tE^\mathrm{I}}^1_1(x)}, \dots}.
\end{equation}
As the determinant of the right tensor factor in \eref{eq:symplectic_matrix} is clearly vanishing, the determinant of $M$ is vanishing and it is not invertible.
Hence the presymplectic form is not symplectic and it is not possible to invert it for a corresponding Poisson bracket for real and imaginary parts of connection and electric field. 

Without adding additional Poisson brackets to the theory by hand, it is not possible to rewrite the selfdual Palatini action in terms of real quantities and we have to use complex Ashtekar variables \cite{ashtekar_lectures_1998}.\footnote{~An approach that adds these missing Poisson brackets and therefore works with real an imaginary parts of the canonical variables is \cite{wieland_complex_2012}.}
This further implies that the phase space of selfdual GR can only consist of functions which are holomorphic in both, $A^i_a$ and $\tE^a_i$, and that the Poisson brackets must be understood with respect to holomorphic derivatives. This is the way we want to proceed in the following. 

There is also an immediate consequence for a possible implementation of reality conditions. These conditions would necessarily involve the basic fields and their complex conjugates, hence they would not be functions in the holomorphic phase space described above. Working in a holomorphic setup excludes the possibility to introduce reality conditions as additional constraints.\footnote{~An example for a phase space that includes the basic fields and their conjugates see \cite{wieland_complex_2012}.}
They can only be implemented by hand, for example as adjointness relations.  

\subsubsection{Spatial metric and triads}
\label{sec:spatial_metric_tiads}
One of the beautiful results of the Ashtekar formulation of gravity is that the electric fields serve as densitised triads for the spatial metric, with respect to an internal three metric. Up to now, we refrained from specifying an internal metric, used to move $\sltwoC$ indices $i=1,2,3$, because no such movements of indices took place. 

Without an internal metric, the natural way of contracting internal indices is the Cartan-Killing metric of $\sltwoC$. As the definition of electric fields includes the inverse isomorphism, the inverse Cartan-Killing metric $k^{ij}$ of $\sltwoC$ can be expressed by its $\laCp$ counterpart $k^{IJKL}$ according to \eref{eq:isomorphism_CKmetrics}. With \eref{eq:selfdual_CKm_projector} we therefore find
\begin{equation}\label{eq:EE_expanded_SigmaSigma}
    \tE^a_i \tE^b_j k^{ij}=(-\sqrt{q})^2 \lc{^a_{cd}}\lc{^b_{ef}}k^{ij}\Pji{IJ}{i}\Pji{KL}{j}\Sigma^{cd}_{IJ}\Sigma^{ef}_{KL}=-\frac{1}{2}q \lc{^a_{cd}}\lc{^b_{ef}} P^{IJKL}\Sigma^{cd}_{IJ}\Sigma^{ef}_{KL}.
\end{equation}
As the the pairs of Lorentz indices are already selfdual, the projection act trivially and just moves one pair of indices. Hence we look at 
\begin{equation}
\begin{aligned}
    &\lc{^a_{cd}}\lc{^b_{ef}}\Sigma^{cd}_{IJ}\Sigma^{efIJ}=\\
    &=\frac{1}{4}\lc{^a_{cd}}\lc{^b_{ef}}\Brb{e_{cI} e_{dJ}-\frac{i}{2}\lc{_{IJKL}}e^K_c e^L_d}\Brb{e^{I}_e e^{J}_f-\frac{i}{2}\lc{^{IJ}_{MN}}e^M_e e^N_f}\\
    &=\frac{1}{4}\lc{^a_{cd}}\lc{^b_{ef}}\Bigl(
    \begin{aligned}[t]
    &e_{cI} e^I_e e_{dJ} e^J_f - \frac{i}{2} e_{cI}e_{dJ}\lc{^{IJ}_{MN}}e^M_e e^N_f \\ &-\frac{i}{2} \lc{_{IJKL}}e^K_c e^L_d e^I_e e^J_f - \frac{1}{4} \lc{_{IJKL}}e^K_c e^L_d \lc{^{IJ}_{MN}}e^M_e e^N_f\Bigr).\end{aligned}
\end{aligned}
\end{equation}
Here the two middle terms vanish since $\lc{_{IJKL}}e^I_a e^J_b e^K_c e^L_d= \lc{^{abcd}}=0$ because of four spatial indices. Contracting the Levi-Civita symbols further yields
\begin{equation}
\begin{aligned}
    &\frac{1}{4}\lc{^a_{cd}}\lc{^b_{ef}}\Brb{
    e_{cI} e^I_e e_{dJ} e^J_f +  \delta_{K[M}\delta_{N]K} e^K_c e^L_d e^M_e e^N_f
    }\\
    &=\frac{1}{4}\lc{^a_{cd}}\lc{^b_{ef}} \rb{q_{ce}q_{df}+q_{ce}q_{df}}\\
    &=\frac{1}{2 q} \lcdp{^{acd}}\lcdp{^{bef}}q_{ce}q_{df}=q^{ab}.
\end{aligned}    
\end{equation}
So indeed this complex electric field encoding the selfdual degrees of freedom contains the spatial metric. Inserting this back in \eref{eq:EE_expanded_SigmaSigma} and solving for the spatial metric yields 
\begin{equation}
    q q^{ab}=-2\tE^a_i \tE^b_j k^{ij}.
\end{equation}
This in turn suggest that the (inverse) internal three metric, related to the spatial metric by densitized triads $\tE^a_i$, is indeed
\begin{equation}
    -2k^{ij}=-2\rb{-\frac{1}{2}}\delta^{ij}=\delta^{ij}.
\end{equation}
The internal metric is therefore given by the Euclidean three metric $\delta_{ij}$, which we will use henceforth. This is in fact remarkable. In the derivation of Ashtekar-Barbero gravity from a real Palatini formulation, time-gauge is used to align the frame fields with the foliation. Consequently the structure group is reduced to $\SUtwo$.
In fixing this gauge, the internal Minkowski metric is reduced to its spatial part, which is exactly the Euclidean metric. Here however, we end up with exactly the same internal three metric, but no time gauge has been performed.
It is easy to see that the split of $\laC$ into its (anti-)selfdual components is incompatible with stabilising a gauge fixing similar to the time gauge. As we discuss the number of physical degrees of freedom later on, it turns out, we do not even need an additional gauge fixing. 

Having established the relation
\begin{equation}\label{eq:densitised_metric}
    q q^{ab}=\tE^a_i \tE^b_j \delta^{ij},
\end{equation}
we see that the electric fields can be interpreted as densitised triads of the spatial manifold, even though they carry spacetime information. This is contained in the components $e^0_a$, that enters via the Plebanski bi-vector.

We can now use \eref{eq:densitised_metric} in order to invert the electric field:
\begin{equation}\label{eq:inverse_electric_field}
    (\tE^{-1})^i_a=\frac{1}{q}q_{ab}\delta^{ij}\tE^b_j=\Et^i_a,
\end{equation}
which is a density of weight $-1$. The spatial metric therefore is given by 
\begin{equation}
    q_{ab}=q \delta_{ij}\Et^i_a\Et^j_b.
\end{equation}
In order to fix the notation, we strip the electric fields off their density weight in order to get the actual triads
\begin{equation}
\begin{aligned}
    \Sigma^a_i&=\frac{1}{\sqrt{q}}\tE^a_i=-\lc{^a_{bc}}\Sigma^{bc}_i=-\lc{^a_{bc}}\rb{2\Pji{0j}{i}e^b_0e^c_j+\Pji{kl}{i}e^b_ke^c_l}\\ &=-\lc{^a_{bc}}\rb{i e^b_0e^c_i+\frac{1}{2}\lc{_i^{kl}}e^b_ke^c_l}.
\end{aligned}    
\end{equation}
Here we used the explicit form of the isomorphism in \eref{eq:explicit_isomorphism} and identified uppercase Lorentz indices that run from $1$ to $3$ with lowercase indices. We see how the spacetime components $e_0^a$ enter the triads in the first term, while the second term looks similar to the result in Ashtekar-Barbero gravity\footnote{~Using the determinant of the three by three matrix $e^a_i$, the second term can be expressed by a single $e^a_i$. However, the appearing determinant is not related to the determinant of the spatial metric.}.

\subsubsection{Constraints}
In this section we look at the three other term in the selfdual Palatini action and want to bring them in their standard form, while expressing them in terms of the canonical variables. 

For the Hamilton constraint part \eref{eq:hamiltonian_piece} of the action, we see that the contraction between Plebanski two-form and curvature is fully projected and we can pull back the entire expression to the spatial manifold. Therefore
\begin{equation}\label{eq_Hamilton_constrain_aux}
\begin{aligned}
    &- \int_\curlyM \dxd{4}{x} \frac{2\lambda\sqrt{q} s_e}{\kappa} N  q_\alpha^\gamma q_\beta^\delta \Sigma_{i}^{\alpha\beta} F^{i}_{\gamma\delta}=\\
    &=-\int_\RR\grad{t}\int_M \grad{^3x}\frac{2\lambda\sqrt{q} s_e}{\kappa} N   \Sigma_{i}^{ab} F^{i}_{ab}\\
    &=\int_\RR\grad{t}\int_M\grad{^3x} \frac{\lambda N s_e}{\kappa}\lc{^{ab}_c}\tE^c_i F^i_{ab}.
\end{aligned}
\end{equation}
In the last step we we inverted the spatially projected Plebanski bi-vector for the electric field:
\begin{equation}
    \Sigma^{ab}_i
=-\frac{1}{2\sqrt{q}}\lc{^{ab}_c}\tE^c_i.
\end{equation}
Next we want to replace $\lc{^{ab}_c}\tE^c_i F^i_{ab}$ by $\lc{_i^{kl}}\tE^a_k\tE^b_l F^i_{ab}$. For the determinant of electric fields, it holds that 
\begin{equation}
    \lcdm{_{abc}}\tE^c_i=\det(\tE)\lc{_{ijk}}\Et^j_a\Et^k_b. 
\end{equation}
We multiply this by another factor $\det(\tE)$ and replace the spatial Levi-Civita density by the corresponding tensor, i.\,e.
\begin{equation}\label{eq:eE_to_eEE_aux}
    \frac{\det(\tE)}{\sqrt{q}}\lc{_{abc}}\tE^c_i=\det(\tE)^2\lc{_{ijk}}\Et^j_a\Et^k_b. 
\end{equation}
The relation \eref{eq:densitised_metric} further allows to relate the determinants of spatial metric $q$ and electric fields $\det(\tE)$. Just taking the determinant on both sides yields
\begin{equation}
    q^2=\det(\tE)^2.
\end{equation}
Completely analogous to \eref{eq:sign_det_e} for tetrads, we introduce $s_E$ as the complex sign of $\det(\tE)$, and write 
\begin{equation}
    \det(\tE)=s_E\sqrt{\det(\tE)^2}.
\end{equation}
Because of this we find 
\begin{equation}
    \frac{\det(\tE)}{\sqrt{q}}=s_E \sqrt{\frac{\det(\tE)^2}{q}}=s_E\sqrt{q}.
\end{equation}

Hence, inverting the inverse triads on the right hand side, \eref{eq:eE_to_eEE_aux} takes the form
\begin{equation}
    \lc{^{ab}_c}\tE^c_i=\frac{1}{\sqrt{q}s_E}\lc{_{i}^{jk}}\tE^a_j\tE^b_k. 
\end{equation}

Inserting this into \eref{eq_Hamilton_constrain_aux} yields
\begin{equation}
    \int_\RR\grad{t}\int_M\grad{^3x}\frac{N\lambda s_e}{\kappa}\lc{^{ab}_c}\tE^c_i F^i_{ab}=\int_\RR\grad{t}\int_M \grad{^3x} \frac{N\lambda s_e}{\sqrt{q} \kappa s_E} \lc{_{k}^{ij}}\tE^a_i\tE^b_j F^k_{ab}.
\end{equation}
Hence we can read off the Hamilton constraint $C$
\begin{equation}\label{eq:Hamilton_constraint} 
    N C = -N\frac{\lambda s_e}{\sqrt{q} \kappa s_E} \lc{_{k}^{ij}}\tE^a_i\tE^b_j F^k_{ab}.
\end{equation}
Comparing this to the Hamilton constraint of GR in Ashtekar-Barbero variables, there are two important aspects to remark. First the constraint here consists only of what is usually called the Euclidean Hamilton constraint. The so called Lorentzian part of the Hamilton constraint of real GR in Ashtekar-Barbero variables is not present. This is the reason why GR in selfdual variables is expected to have simpler dynamics. 
The second aspect is the presence of $s_e/s_E$, which encodes the the orientation of both triads and tetrads, which are complex quantities. Therefore right from the beginning it is not clear what choosing an orientation of these would actually mean.

We go on with the diffeomorphism constraint part of the action, i.\,e. \eref{eq:diffeo_piece}. Here we see that 
\begin{equation}
    N^\delta=N^\alpha\delta^\delta_\alpha=N^\alpha q^\delta_\alpha-N^\alpha n_\alpha n^\delta=N^\alpha q^\delta_\alpha
\end{equation}
as $N^\alpha n_\alpha=0$. The shift vector is indeed spatially projected. Hence, using again spacetime selfduality and performing the pullback to the spatial manifold 
\begin{equation}
\begin{aligned}
    &-\int_\curlyM \dxd{4}{x}    \frac{4\lambda\sqrt{q} s_e}{\kappa} N^\delta  q_\alpha^\gamma n_\beta \Sigma_{i}^{\alpha\beta} F^{i}_{\gamma\delta}\\
    &=- \int_\curlyM \dxd{4}{x}\frac{4\lambda\sqrt{q} s_e}{\kappa} N^\lambda q^\delta_\lambda \frac{i}{2s_e}n_\beta \lc{^{\alpha\beta}_{\rho\sigma}} \Sigma_{i}^{\rho\sigma} F^{i}_{\gamma\delta}\\
    &=\int_\RR\grad{t}\int_M \grad{^3x} \frac{2\lambda\sqrt{q} i}{\kappa} N^a  \lc{^{d}_{bc}} \Sigma_{i}^{bc} F^{i}_{da}\\
    &=\int_\RR\grad{t}\int_M \grad{^3x} \frac{2\lambda i}{\kappa} N^a  \tE^b_i F^{i}_{ab},
\end{aligned}    
\end{equation}
where we again expressed everything by curvature and electric field. The diffeomorphism constraint $C_a$ now reads
\begin{equation}\label{eq:diffeo_constraint}
    N^a C_a = \frac{2\lambda}{i\kappa} N^a  \tE^b_i F^{i}_{ab}.
\end{equation}
Likewise to the canonical pair, the dependence on the orientation of tetrads disappeared. 

Finally we work out the Gau{\ss} constraint from \eref{eq:gauss_piece}. Once again employing spacetime selfduality and using that the open indices of $n_\beta\lc{^{\beta\alpha}_{\rho\sigma}}$ are already projected (making additional projections obsolete), we find
\begin{equation}
    \begin{aligned}
        &\int_\curlyM \dxd{4}{x} \frac{4\lambda\sqrt{q} s_e}{\kappa}  q_\alpha^\gamma n_\beta \Sigma_{i}^{\alpha\beta} \mathop{D_\gamma} (t^\delta A^i_\delta)=\\
        &=-\int_\curlyM \dxd{4}{x} \frac{4\lambda\sqrt{q} s_e}{\kappa} \frac{i}{2 s_e}  n_\beta \lc{^{\beta\gamma}_{\rho\sigma}} \Sigma_{i}^{\rho\sigma} \mathop{D_\gamma} (t^\delta A^i_\delta)\\
        &=\int_\curlyM \dxd{4}{x} \frac{2\lambda\sqrt{q} }{i\kappa}  n_\beta \lc{^{\beta\gamma}_{\rho\sigma}} \Sigma_{i}^{\rho\sigma} \mathop{D_\gamma} (t^\delta A^i_\delta).
    \end{aligned}
\end{equation}
Next we use the Leibniz rule on the covariant derivative. This produces the term 
\begin{equation}
    \mathop{D_\gamma} \rb{
    \frac{2\lambda\sqrt{q} }{i\kappa}  n_\beta \lc{^{\beta\gamma}_{\rho\sigma}} \Sigma_{i}^{\rho\sigma} t^\delta A^i_\delta},
\end{equation}
which is the covariant derivative of a vector density of weight one. The covariant derivative can therefore be replaced by a partial derivative, we obtain a surface term and drop it. Hence we are left with 
\begin{equation}
\begin{aligned}
    &-\int_\curlyM \dxd{4}{x} (t^\delta A^i_\delta)\mathop{D_\gamma}\rb{\frac{2\lambda\sqrt{q} }{i\kappa}  n_\beta \lc{^{\beta\gamma}_{\rho\sigma}} \Sigma_{i}^{\rho\sigma} t^\delta A^i_\delta}=\\
    &= \int_\RR\grad{t}\int_M \grad{^3x}  (t^\delta A^i_\delta)\mathop{D_a}\rb{-\frac{2\lambda\sqrt{q} }{i\kappa}  \lc{^{a}_{bc}} \Sigma_{i}^{bc}} \\
    &= \int_\RR\grad{t}\int_M \grad{^3x}  (t^\delta A^i_\delta) \frac{2\lambda }{i\kappa} \mathop{D_a}(\tE^a_i),
\end{aligned}    
\end{equation}
where we performed the pullback and introduced the electric field. The derivative $D_a$ is the spatial projection of $D_\alpha$ and hence the covariant derivative with respect to $A_a^i$. With $(A\cdot t)^i:=t^\delta A^i_\delta$ The Gau{\ss} constraint $G_i$ reads
\begin{equation}
    (A\cdot t)^i G_i = - \frac{2\lambda }{i\kappa} \mathop{D_a}(\tE^a_i).
\end{equation}
It is easy to see that this generates indeed $\SLtwoC$ gauge transformations. 

Unlike the Hamilton constraint, the diffeomorphism and Gau{\ss} constraints are of exactly the same form as in the real theory, cf. \cite{ashtekar_background_2004,thiemann_modern_2007}.

Putting all pieces of the action together, we therefore have 
\begin{equation}
    S=\int_\RR\grad{t}\int_M \grad{^3x} \rb{\frac{2\lambda}{\kappa i} \tE^a_i\mathop{\curlyL_t}(A^i_a)-\rb{NC+N^aC_a+(A\cdot t)^i G_i}}.
\end{equation}

Once again, we can comment on the consequences of rescaling the action by $\lambda$. As we just established, all the constraints are rescaled by a factor $\lambda$, while the Poisson bracket is rescaled by a factor $1\lambda$. Therefore the equations of motion obtained from this action are unchanged as
\begin{equation}
    \PB{\slotdot}{C_{(\lambda)}}_{(\lambda)}=\frac{1}{\lambda}\PB{\slotdot}{\lambda C}=\PB{\slotdot}{C},
\end{equation}
for all types of constraints. Similarly, the form of the hypersurface deformation algebra is unaffected. 

Having established all the constraints, we now want to see how many physical degrees of freedom from complexified GR are left. We will do this on the $\sltwoC$ level of the formulation only and count complex degrees of freedom. The configuration variable $A^i_a$ encodes 9 complex degrees of freedom. The Hamiltonian constraints reduces this number by one, diffeomorphism and Gau{\ss} constraint reduce by 3 respectively. Therefore we are left with 2 complex degrees of freedom. Expecting the reality conditions to halve this again, we end up with two physical and real degrees of freedom for gravity.
As mentioned earlier, we do not need an additional gauge fixing (time gauge) in order to obtain the correct number of degrees of freedom.

\section{Reality conditions}\label{sec:reality_conditions}
As already mentioned in the introduction, the reality conditions suggested throughout the literature, cf. \cite{ashtekar_new_1989,ashtekar_lectures_1998,ashtekar_gravitational_2021} are reality of the spatial metric and reality of its time evolution. The first reality condition (RCI) therefore is
\begin{equation}\label{eq:RCI_metric}
    q_{ab}\in \RR.
\end{equation}
With the dot indicating the derivative with respect to the foliation, the second reality condition (RCII) is
\begin{equation}\label{eq:RCII_metric}
    \dot{q}_{ab}\in \RR.
\end{equation}

Various formulations of the reality conditions have been considered in the literature, most recently in \cite{ashtekar_gravitational_2021}, which also includes the treatment of the Euclidean signature. In the following, we will consider the case of general $\lambda$ and give rather detailed derivations.

The reality conditions \eref{eq:RCI_metric} and $\eref{eq:RCII_metric}$ have to be translated into conditions for the selfdual variables $E$ and $A$. It can then be shown that under these reality conditions and in the absence of matter, the theory reduces to real Einstein gravity. 
Starting point for this is the densitised inverse spatial metric \eref{eq:densitised_metric}.

\subsection{First reality condition}\label{sec:RCI}

Before we start with the first reality condition we present the concept of reality and imaginarity up to gauge rotations introduced in \cite{ashtekar_gravitational_2021}. 

We call objects real up to a gauge rotations if their imaginary part can be removed by a rotation. Similarly we call objects imaginary up to a gauge rotations if their real part can be removed by a rotation. In our case here rotations mean complex orthogonal transformation of $\mathrm{O}(3,\CC)$.

With this at hand we can define reality of the spatial metric with respect to the electric field.

\begin{proposition}
The following statements are equivalent:\begin{enumerate}[label=(\roman*)]
    \item \label{thm:E_real} $\tE^a_i$ is real modulo gauge rotations, 
    \item \label{thm:qq_posdef} $qq^{ab}=\delta^{ij}\tE^a_i\tE^b_j$ is real and positive definite, 
    \item \label{thm:q_real} $q^{ab}$ is real.
\end{enumerate}
\end{proposition}

\begin{proof}
\ref{thm:E_real} $\Rightarrow$ \ref{thm:qq_posdef}: 
Let the electric field be real modulo gauge rotations, i.\,e. $\tE^a_i=\T{R}{^j_i}{\tE_R}^a_j$ with ${\tE_R}^a_j\in\RR$. Since the Euclidean metric is invariant under rotations, we find 
\begin{equation}
    \begin{aligned}
        \tE^a_i\tE^b_j\delta^{ij}&= {\tE_R}^a_m{\tE_R}^b_n  \T{R}{^m_i}\T{R}{^n_j}\delta^{ij}\\
        &={\tE_R}^a_m{\tE_R}^b_n \delta^{mn}.
    \end{aligned}
\end{equation}
This is positive definite because $\delta_{ij}$ is positive definite. 

\ref{thm:qq_posdef} $\Rightarrow$ \ref{thm:q_real}: 
Let $qq^{ab}$ be positive definite. Then the determinant 
\begin{equation}
    \det(qq^{ab})=q^3\frac{1}{q}=q^2
\end{equation}
is strictly positive. Therefore, decomposing $q=q_R+iq_I$ with $q_R,q_I\in\RR$ into real and imaginary part, 
\begin{equation}
    0<q^2={q_R}^2-{q_I}^2+2iq_R q_I.
\end{equation}
This implies $q_I=0$, i.\,e. $q\in\RR$. Hence $q^{ab}$ is real. 

\ref{thm:q_real} $\Rightarrow$ \ref{thm:E_real}:
Let $q^{ab}\in\RR$, hence $\exists$ ${\tE_R}^a_i\in\RR$ such that $qq^{ab}={\tE_R}^a_i{\tE_R}^b_j\delta^{ij}$. For any possibly complex ${\tE'}^a_i$ with ${\tE'}^a_i{\tE'}^b_j\delta^{ij}=qq^{ab}$ it therefore must hold  true that
\begin{equation}
\begin{aligned}
    &{\tE'}^a_i{\tE'}^b_j\delta^{ij}={\tE_R}^a_i{\tE_R}^b_j\delta^{ij}\\
    \Leftrightarrow &{\tE'}^a_i {\tE'}^b_j{(\tE^{-1}_R)}_a^m {(\tE^{-1}_R)}_b^n\delta^{ij}= \delta^{mn}\\
    \Leftrightarrow & (\tE'\tE^{-1}_R)^m_i (\tE'\tE^{-1}_R)^n_j\delta^{ij}=\delta^{mn}.
\end{aligned}    
\end{equation}
This implies $(\tE'\tE^{-1}_R)$ is a rotation, which leaves $\delta_{ij}$ invariant. Hence 
\begin{equation}
    \tE^a_i=\T{R}{^j_i}{\tE_R}^a_j,
\end{equation}
i.\,e. $\tE^a_i$ is real up to a gauge rotation. 
\end{proof}

Dropping positive definiteness in \ref{thm:qq_posdef} would allow for a second solution, namely $q=iq_I$. This however would imply a purely imaginary metric. 

When neglecting the imaginary part that is removable by a gauge rotation, the reality of the electric field $\tE^a_i$, according to the first reality condition, we see that it corresponds to an infinitesimal rotation.

\subsection{Second reality condition}\label{sec:RCII}
Before we further analyse the second reality condition, we need the complex analogue of the extrinsic curvature $K^i_a$. To this end we derive the covariant derivative that is compatible with the electric field. Keeping in mind that $\tE^a_i$ is a density of weight 1, we have
\begin{equation}
    \mathop{\curlyD_a}\tE^b_j=\partial_a\tE^b_j+\T{\Gamma}{^b_{ac}}\tE^c_j-\T{\Gamma}{^c_{ca}}\tE^b_j+\lc{_{ji}^{k}}\Gamma^i_a\tE^b_k=0.
\end{equation}
We want to solve this for the spin connection $\Gamma^i_a$. Thus we multiply by the inverse electric field $\Et^m_b$ and contract with the structure constant. This yields 
\begin{equation}
    \lc{^{jk}_m}\Et^m_b\partial_a\tE^b_j+\lc{^{jk}_m}\T{\Gamma}{^b_{ac}}\Et^m_b\tE^c_j-\lc{^{jk}_j}\T{\Gamma}{^c_{ca}}+\lc{^{jk}_m}\lc{_{ji}^{m}}\Gamma^i_a=0.
\end{equation}
Since $\lc{^{jk}_m}\lc{_{ji}^{m}}=2\delta^k_i$ and $\lc{^{jk}_m}=0$,  the contribution for the density weight in the covariant derivative does not affect the spin connection. This yields
\begin{equation}\label{eq:spin_connection_Gamma}
    \Gamma^k_a=\frac{1}{2}\lc{^{kj}_m}\rb{\Et^m_b\partial_a\tE^b_j+\T{\Gamma}{^b_{ac}}\Et^m_b\tE^c_j}.
\end{equation}
Assuming the first reality condition holds, we see that character of the spin connection is purely real. This is because the reality of the electric field is transferred to its inverse. The Christoffel symbols are real as they are a function of the then real internal metric. Already here we see that the spin connection generates rotations of the compact subgroup once the reality conditions hold.\footnote{~We should remark that is statement is highly basis dependent. If we would  for example choose the selfadjoint basis of $\sltwoC$, the structure constants would be $i\lc{_{ij}^k}$ and hence the spin connection would be purely imaginary. Nevertheless it would generate rotations in this basis.}

As it is well known in the literature, e.\,g. \cite{ashtekar_new_1986,ashtekar_lectures_1998,ashtekar_gravitational_2021}, the extrinsic curvature is defined as the difference between covariant derivatives with respect to the Ashtekar connection and the spin connection above:
\begin{equation}\label{eq:cov_dev_difference}
    \rb{\mathop{D_a}-\mathop{\curlyD_a}}v_i=\sce{_{ij}^k}\rb{A^j_a-\Gamma^j_a}v_k.
\end{equation}
It is therefore the difference of the connections, namely
\begin{equation}
    K^j_a=\frac{1}{i}\rb{A^j_a-\Gamma^j_a}.
\end{equation}
This also gives rise to a decomposition of the Ashtekar connection
\begin{equation}
    A^j_a=\Gamma^j_a+iK^j_a
\end{equation}
into spin connection and extrinsic curvature. We will see that under the second reality condition this is a split into real and imaginary part. 

Further we need to clarify the derivative with respect to the foliation. The spatial metric is a purely spatial object, as it does not carry any internal indices. It is invariant under diffeomorphisms and internal gauge transformations. Therefore when computing the Poisson bracket with the full Hamiltonian consisting of Hamilton, diffeomorphism and Gau{\ss} constraint, only the contribution of the Hamilton constriant is non-vanishing \cite{ashtekar_new_1989}. Hence in the following the dot indicates the Poisson bracket with the Hamilton constraint. 

We can now relate the reality of the derivative of the spatial metric to the extrinsic curvature. See also \cite{ashtekar_gravitational_2021}.

\begin{proposition}
Let the first reality condition hold. Then the following statements are equivalent:
\begin{enumerate}[label=(\roman*)]
\item \label{thm:K_real} $K^i_a$ is, up to a gauge rotation, the boost part of the Ashtekar connection $A^i_a$,
\item \label{thm:qqdot_real} $(qq^{ab})^{\boldsymbol{\cdot}}=(\delta^{ij}\tE^a_i\tE^b_j)^{\boldsymbol{\cdot}}$ is real,
\item \label{thm:qdot_real} $\dot{q}^{ab}$ is real.
\end{enumerate}
\end{proposition}

\begin{proof}
\ref{thm:qdot_real} $\Rightarrow$ \ref{thm:qqdot_real}: With RCI $q_{ab}$ and $q$ are real.  We compute the derivative of the metric determinant,
\begin{equation}
    \dot{q}=\frac{3}{3!}\lcdp{^{abc}}\lcdp{^{a'b'c'}}q_{aa'}q_{bb'}\dot{q}_{cc'}.
\end{equation}
Since $\dot{q}_{ab}\in\RR$ by assumption, $\dot{q}\in\RR$. Hence 
\begin{equation}
    (qq^{ab})^{\boldsymbol{\cdot}}=\dot{q}q^{ab}+q\dot{q}^{ab}
\end{equation}
is real.

\ref{thm:qqdot_real} $\Rightarrow$ \ref{thm:K_real}: We start with expanding the derivative. This yields 
\begin{equation}
    (qq^{ab})^{\boldsymbol{\cdot}}=\delta^{ij}(\tE^a_i\tE^b_j)^{\boldsymbol{\cdot}}=2\delta^{ij}\tE^{(a}_i\dot{\tE}^{b)}_j.
\end{equation}
Hence we need to compute the derivative of the electric field. Thus with the Hamilton constraint in \eref{eq:Hamilton_constraint}, we need to work out
\begin{equation}
\begin{aligned}
    \dot{\tE}^a_j(x)&=\PB{\tE^a_j(x)}{C[N]}=\\
    &=-\int_M\!\grad{^3y}\rb{N\frac{\lambda s_e}{\sqrt{q} \kappa s_E} \lc{_k^{mn}}\tE^c_m\tE^b_n}(y) \PB{\tE^a_j(x)}{F^k_{cb}(y)}.
\end{aligned}    
\end{equation}
Here the only non-vanishing contribution from the Hamilton constraint is the Poisson bracket with the curvature
\begin{equation}\label{eq:curvature_of_A}
    F^k_{ab}=2\partial_{[a}A^k_{b]}+\lc{_{mn}^k}A^m_a A^n_b.
\end{equation}
We therefore first compute
\begin{equation}
\begin{aligned}
    &\lc{_{mn}^k}\PB{\tE^a_j(x)}{A^m_a(y)A^n_b(y)}=-\frac{i\kappa}{2\lambda}\lc{_{mn}^k}\rb{\delta^c_a\delta^m_j A^n_b(y)+\delta^c_b\delta^n_j A^m_a(y)}\delta(x,y)\\
    &=-\frac{i\kappa}{2\lambda}\lc{_{jn}^k}\rb{\delta^c_a A^n_b(y)+\delta^c_b A^n_a(y)}\delta(x,y)\\
    &=-i\frac{\kappa}{\lambda}\lc{_{jn}^k}\delta^c_{[a} A^n_{b]}(y)\delta(x,y).
\end{aligned}    
\end{equation}
With the $\partial_{[a}A^k_{b]}$ piece of \ref{eq:curvature_of_A} we deal as follows. The anti-symmetrisation is obsolete in the presence of the structure constants. The derivative is with respect to $y$ and therefore acts on $y$-dependent objects only and we can move it past $\tE^a_j(x)$. We can therefore perform an integration by parts, where we drop the surface term, and evaluate the Poisson bracket. Further introducing $\Nt=N/\sqrt{q}$ we find 
\begin{equation}
    \begin{aligned}
    \dot{\tE}^a_j(x)&=-\frac{\lambda}{\kappa }\int_M\!\grad{^3y}\Biggl[\begin{aligned}[t]
    &-\partial_c\Brb{\Nt\frac{s_e}{s_E} \lc{_k^{mn}}\tE^c_m\tE^b_n}(y) 2\Brb{-\frac{i\kappa}{2\lambda}}\delta^a_b\delta^k_j\\ &+ \Brb{\Nt\frac{s_e}{s_E} \lc{_m^{mn}}\tE^c_m\tE^b_n}(y)\rb{-\frac{i\kappa}{\lambda}\lc{_{jl}^k}\delta^a_{[c} A^l_{b]}(y)}\Biggr]\delta(x,y)
    \end{aligned}\\
    &= i\rb{
    -\lc{_j^{mn}}\partial_c\Brb{\Nt\frac{s_e}{s_E} \tE^c_m\tE^a_n}+ \Nt\frac{s_e}{s_E} \lc{_k^{mn}}\lc{_{jl}^k}\tE^a_m\tE^b_n A^l_b}(x)\\
    &=-i\rb{
    \lc{_j^{mn}}\partial_c\Brb{\Nt\frac{s_e}{s_E} \tE^c_m\tE^a_n}- 2\Nt\frac{s_e}{s_E} \tE^{[a}_j\tE^{c]}_l A^l_c}(x),
\end{aligned}  
\end{equation}
where, in the second to last step, we performed the integration. We note that, as expected, the $\lambda$-dependence drops out in the Poisson bracket with a constraint. 
We omit the $x$-dependence of the electric field again. 
Using that in fact 
\begin{equation}
    \lc{_j^{mn}}\rb{\lc{_{mm'}^{n'}}A^{m'}_c\tE^{[a}_{n'}\tE^{c]}_n+
    \lc{_{nm'}^{n'}}A^{m'}_c\tE^{[a}_{m}\tE^{c]}_{n'}}=-2A^l_c\tE^{[a}_j\tE^{c]}_l,
\end{equation}
we can write this in terms of a covariant derivative, i.\,e. 
\begin{equation}
    \begin{aligned}
    \dot{\tE}^a_j&=-i\rb{
    \lc{_j^{mn}}\partial_c\Brb{\Nt\frac{s_e}{s_E} \tE^c_m\tE^a_n}- 2\Nt\frac{s_e}{s_E} \tE^{[a}_m\tE^{c]}_n A^l_c}\\
    &=-i\lc{_j^{mn}}\mathop{D_c}\rb{\Nt\frac{s_e}{s_E} \tE^c_m\tE^a_n}.
    \end{aligned}
\end{equation}

We can now go on with the expression we are actually interested in:
\begin{equation}
\begin{aligned}
     (qq^{ab})^{\boldsymbol{\cdot}}&=-2i\delta^{ij}\lc{_j^{kl}}\tE^{(a}_i\mathop{D_c}\brb{\Nt\frac{s_e}{s_E} \tE^{b)}_k\tE^c_l}\\
             &=-2i\lc{^{jkl}}\tE^{(a}_j\rb{\Nt\frac{s_e}{s_E} \mathop{D_c}\brb{\tE^{b)}_k\tE^c_l} +\tE^{b)}_k\tE^c_l\mathop{D_c}\brb{\Nt\frac{s_e}{s_E}}}\\
             &=-2i\Nt\frac{s_e}{s_E} \lc{^{jkl}}\tE^{(a}_j\rb{\brb{\mathop{D_c}\tE^{b)}_k}\tE^c_l+\tE^{b)}_k\brb{\mathop{D_c}\tE^c_l}}\\
             &=-2i\Nt\frac{s_e}{s_E} \lc{^{jkl}}\tE^{(a}_j\rb{\tE^{b)}_k\lc{_{km}^n}\rb{A^m_c-\Gamma^m_c}\tE^c_l+\tE^{b)}_k\lc{_{lm}^n}\rb{A^m_c-\Gamma^m_c}\tE^c_l}\\
             &=2\Nt\frac{s_e}{s_E} \lc{^{kl}_j}\tE^{(a}_j\rb{\tE^{b)}_k\lc{_{km}^n}K^m_c\tE^{b)}_k\tE^c_l+\tE^{b)}_k\lc{_{lm}^n}K^m_c\tE^c_l}\\
             &=2\Nt\frac{s_e}{s_E} \tE^{(a}_j\rb{\tE^{b)}_k \rb{-\delta_{jm}\delta^{ln}+\delta^n_j\delta^l_m} K^m_c\tE^c_l+\tE^{b)}_k \rb{\delta_{jm}\delta^{kn}-\delta^n_j\delta^k_m}K^m_c\tE^c_l}\\
             &=2\Nt\frac{s_e}{s_E} \tE^{(a}_j\rb{-\tE^{b)}_n K_{jc} \tE^{cn} + \tE^{b)}_j K^l_c \tE^c_l  +\tE^{b)}_n K_{jc} \tE^{cn} - \tE^{b)}_k K^k_c \tE^c_j}\\
             &=2\Nt\frac{s_e}{s_E} \rb{qq^{ab}\tE^c_l K^l_c - q\tE^{(b}_m q^{a)c} K^k_c }
\end{aligned}     
\end{equation}
The second term in the second line vanishes because of a contraction of symmetric with anti-symmetric indices. From the third to the fourth line we replaced the covariant derivatives $D_c$ with $D_c-\curlyD_c$ and used \eref{eq:cov_dev_difference}. The difference of connections was then replaced by $iK^m_c$ and we contracted the structure constants. From the penultimate to the last line, we contracted electric fields and the internal metric to the densitised spatial metric. Absorbing the metric determinant in the lapse, this yields
\begin{equation}\label{eq:RCII_final}
    (qq^{ab})^{\boldsymbol{\cdot}}=2\tN\frac{s_e}{s_E} \rb{\tE^c_m q^{ab}-\tE^{(a}_m q^{b)c} } K^m_c.
\end{equation}
By assumption the first reality condition holds, so the spatial metric and the densitised lapse are real. Further $\tE^c_m$ is real up to a gauge rotation. This implies that indeed $K^m_c$ is real up to the corresponding inverse gauge rotation.
Therefore without referring to a specific basis, $(qq^{ab})^{\boldsymbol{\cdot}}\in\RR$ implies that up to a gauge rotation $K^i_a$ is the boost part of the Ashtekar connection $A^i_c$.\footnote{~This statement as well is highly basis dependent. Similar to the spin connection described earlier, becoming imaginary in the selfadjoint basis of $\sltwoC$, the Levi-Civita symbols in the calculation above are actually the structure constants and change with respect to the basis. One then has to carefully work out index positions and inverse structure constants in order to get contraction right. The result then is an additional factor of $i$ in front of \eref{eq:RCII_final}. Therefore the extrinsic curvature would have to be imaginary up to a gauge rotation, but would still be the boost part of the Ashtekar connection.}

\ref{thm:K_real} $\Rightarrow$ \ref{thm:qdot_real}:
Here we need the derivative of the metric:
\begin{equation}
\dot{q}^{ab}=\rb{\frac{1}{q}\delta^{ij}\tE^a_i\tE^b_j}^{\cdot}=\rb{\frac{1}{s_E \det(\tE)}}^{\cdot} qq^{ab}+\frac{1}{q}\rb{qq^{ab}}^{\cdot}.
\end{equation}
The only thing not known to be real at this point is 
\begin{equation}\label{eq:derivatives_sqdetetasEdetE}
\rb{\frac{1}{s_E \det(\tE)}}^{\cdot}=-\rb{\frac{\dot{s_E}}{{s_E}^2\det(\tE)}+\frac{1}{s_E}{\rb{\frac{1}{\det(\tE)}}}^\cdot}.
\end{equation}
Both $s_E$ and $\det(\tE)$ are functions of $\tE^a_j$ only. We can therefore compute the Poisson bracket with the Hamilton constraint as follows:
\begin{equation}
\PB{s_E}{C[N]}=-\frac{i\kappa}{2\lambda}\int_M\grad{^3z}\frac{\delta s_E(x)}{\delta\tE^a_j(z)}\frac{\delta C[N]}{\delta A_a^j(z)}.
\end{equation}
Here we already know the derivative of the Hamilton constraint, namely
\begin{equation}
-\frac{i\kappa}{2\lambda}\frac{\delta C[N]}{\delta A_a^j(z)}=\dot{\tE}^a_j(z)=-i\lc{_j^{mn}}\mathop{D_c}\rb{\Nt\frac{s_e}{s_E} \tE^c_m\tE^a_n}.
\end{equation}
The derivative of the sign is given by 
\begin{equation}
\frac{\delta s_E(x)}{\delta\tE^a_j(z)}=-s_E\Et^j_a(x)\delta(x,z).
\end{equation}
So upon performing the $z$ integration and again replacing the covariant derivative by $D-\curlyD$ and hence introducing $K^j_a$ in the expressions, we find 
\begin{equation}
\dot{s_E}=2 s_e \Nt K^m_c\tE^c_m.
\end{equation}
Similarly we need the Poisson bracket of the determinant:
\begin{equation}
\PB{\frac{1}{\det(\tE)}}{C[N]}=-\frac{i\kappa}{2\lambda}\int_M\grad{^3z}\rb{\frac{\delta}{\delta\tE^a_j(z)}\frac{1}{\det(\tE)(x)}}\frac{\delta C[N]}{\delta A_a^j(z)}.
\end{equation}
The derivative of the inverse determinant is given by
\begin{equation}
\rb{\frac{\delta}{\delta\tE^a_j(z)}\frac{1}{\det(\tE)(x)}}=-\frac{1}{\det(\tE)}\Et^j_a(x)\delta(x,z).
\end{equation}
Again putting this together, performing the $z$ integration and introducing $K^j_a$ yields
\begin{equation}
{\rb{\frac{1}{\det(\tE)}}}^\cdot=
\frac{2 s_e\Nt}{\det(\tE)s_E}K^m_c\tE^c_m
\end{equation}
Inserting all of this in \eref{eq:derivatives_sqdetetasEdetE} we therefore find the rather compact expression
\begin{equation}
\rb{\frac{1}{s_E \det(\tE)}}^{\cdot}=-\frac{4s_e\Nt}{\det(\tE)}K^m_c\tE^c_m.
\end{equation}
As possible gauge rotations cancel in the contraction of electric field and extrinsic curvature, which are both real by assumption, this is real valued.
Consequently, $\dot{q}^{ab}$ is real. This completes the proof of the proposition. 
\end{proof}

\subsection{Recovering real gravity in ADM formulation}
\label{sec:recovering_ADM}
As a last step we want to show that indeed both values of $\lambda$ reproduce a real formulation of gravity, once the reality conditions hold. However, we cannot do this as long as we work in a formulation in terms of $\tE^a_j$ and $A^j_a$ as the reality conditions for are not implementable classically. To this end we need to reduce the holomorphic formulation in terms of the complex Ashtekar connection and the electric field to a formulation in terms of the extrinsic curvature and the electric field. This allows to compare the Poisson relation and the constraints to the (real) ADM formulation of general relativity in these variables, as e.\,g. presented in \cite{thiemann_modern_2007}. 

A similar analysis is performed when it is shown that real gravity in terms of Ashtekar-Barbero variables is equivalent to the ADM formulation, as long as the Gau{\ss} constraint holds. This is e.\,g. also shown in \cite{thiemann_modern_2007}. We perform essentially the same steps but have in addition to carefully work out which quantities are complex. 

Before we actually start, we need to express the Gau{\ss} constraint by extrinsic curvature and electric field. The covariant derivative with respect to the Ashtekar connection can be written as 
\begin{equation}
    D_a v_i=\curlyD_a v_a + i\lc{_{ij}^k} K^j_a v_k.
\end{equation}
Now using that $\curlyD_a$ is compatible with $\tE^b_j$, we hence find 
\begin{equation}\label{eq:Gauss_constraint_E_K}
    G_i=\frac{2\lambda}{\kappa i} D_a\tE^a_i = \frac{2\lambda}{\kappa}\lc{_{ij}^k} K^j_a\tE^a_k .
\end{equation}
Hence in the following we express the Poisson relation and remaining constraints in terms of the ADM variables up to this Gau{\ss} constraint.

\subsubsection{Poisson relation}
Simply solving the Ashtekar connection for the extrinsic curvature, we can compute the Poisson bracket between $K^i_a$ and the electric field. We find 
\begin{equation}
    \PB{K^i_a(x)}{\tE^b_j(y)}=\frac{1}{i}\PB{A^i_a(x)-\Gamma^i_a(x)}{\tE^b_j(y)}=\frac{\kappa}{2\lambda}\delta^i_j\delta^a_b\delta(x,y),
\end{equation}
since $\PB{\Gamma^i_a(x)}{\tE^b_j(y)}=0$. For both $K^i_a$ and $\tE^b_j$ real, and  $\lambda=1$ this perfectly reproduces the ADM Poisson relation
\begin{equation}
    \PB{K^i_a(x)}{\tE^b_j(y)}=\frac{\kappa}{2}\delta^i_j\delta^a_b\delta(x,y),
\end{equation}
Further, as the right hand side -- in form of the Kronecker-deltas -- is invariant under orthogonal transformations, any any gauge rotation present is removable and the Poisson relation would be equivalent to the ADM relation.

For $\lambda=i$ instead, we have -- independent of the presence of gauge rotations -- an additional factor of i present:
\begin{equation}
    \PB{K^i_a(x)}{\tE^b_j(y)}=\frac{\kappa}{2i}\delta^i_j\delta^a_b\delta(x,y).
\end{equation}
Once we work out the reality of the constraints, it will be clear that this change to the ADM relation is consistent with a change of ADM constraints in this case and nevertheless reproduces the dynamics of ADM.

\subsubsection{Diffeomorphism constraint}\label{sec:ADM_diffeo}
We recall from section \ref{sec:SelfdualPalatiniAction} the diffeomorphism constraint \eref{eq:diffeo_constraint} and the explicit expression for the curvature of the Ashtekar connection \eref{eq:curvature_of_A}. In the latter we replace $A^i_a$ by $\Gamma^i_a+i K^i_a$ and therefore find 
\begin{equation}
\begin{aligned}
    C_a &= \frac{2\lambda}{i\kappa} \tE^b_i F^{i}_{ab} \\
    &= \frac{2\lambda}{i\kappa} \tE^b_i \rb{2 \partial_{[a}(\Gamma^i_{b]}+i K^i_{b]}) + \lc{_{kl}^i} \rb{\Gamma^k_a\Gamma^l_b+i\Gamma^k_aK^l_b+iK^k_a\Gamma^l_b-K^l_aK^k_b}}\\
    &=\frac{2\lambda}{i\kappa} \rb{ \tE^b_i R^i_{ab} +2i \tE^b_i \partial_{[a}K^i_{b]} + i \lc{_{kl}^i} \rb{\Gamma^k_a K^l_b-\Gamma^k_b\Gamma^l_a}+K^i_b G_i }.
\end{aligned}    
\end{equation}
Here we identified the Gau{\ss} constraint, used the antisymmetry of the structure constant and combined all terms with the spin connection only in its curvature 
\begin{equation}
    R^i_{ab}=2 \partial_{[a}\Gamma^i_{b]}+\lc{_{kl}^i} \Gamma^k_a\Gamma^l_b.
\end{equation}
Similar to a real formulation, and e.\,g. as shown in \cite{thiemann_modern_2007}, the contraction $\tE^b_i R^i_{ab}$ vanishes because of the Bianchi identity. Realising that the antisymmetrisation of spatial indices would cause Christoffel symbols to vanish, we can introduce the covariant derivative compatible with $\tE^a_i$ in the expression. Because of the compatibility, we can further move the electric field inside the derivative. Hence, supposing the Gau{\ss} constraint is vanishing, we find
\begin{equation}\label{eq:diffeo_ADM}
    C_a=\frac{4\lambda}{\kappa} \curlyD_{[a}(K^i_{b]}\tE^b_i)=-\frac{2\lambda}{\kappa}\curlyD_b\rb{K^i_a\tE^b_i-\delta^b_aK^i_c\tE^c_i}.
\end{equation}
This has exactly the form of the ADM diffeomorphism constraint and therefore perfectly reproduces it for $\lambda=1$. Present gauge rotations cancel in the contraction of electric field and extrinsic curvature. In the case $\lambda=i$, the real constraint is rescaled by $i$, as expected. We come back to this after working out the Hamilton constraint.

\subsubsection{Hamilton constraint}
Eventually we express the Hamilton constraint  \eref{eq:Hamilton_constraint} by electric field and extrinsic curvature. From our treatment of the diffeomorphism constraint we already known what the curvature of the Ashtekar connection turns into. Only the contractions with electric fields are different. Hence 
\begin{equation}
\begin{aligned}
    C &= -\frac{\lambda s_e}{\sqrt{q} \kappa s_E} \lc{_{k}^{ij}}\tE^a_i\tE^b_j F^k_{ab}\\
    &= -\frac{\lambda s_e}{\sqrt{q} \kappa s_E} \lc{_{k}^{ij}}\tE^a_i\tE^b_j \rb{R^k_{ab}+2i \curlyD_{[a}K^k_{b]}- \lc{_{mn}^k}K^m_a K^n_b }.
\end{aligned}    
\end{equation}
We look individually at the three contributions. 

The first term we consider is the one including the curvature of $\Gamma^i_a$. A detailed calculation reveals that this contribution is in fact proportional to the Ricci curvature $R$ of the spin connection. The result is
\begin{equation}\label{eq:Ricci_scalar}
    \lc{_{k}^{ij}}\tE^a_i\tE^b_j R^k_{ab} = -qR.
\end{equation}

Next we show that the second term is related to the Gau{\ss} constraint. In order to realise this, we write out the antisymmetrisation in the covariant derivative and use for each term the corresponding electric field to form the Gau{\ss} constraint according to \eref{eq:Gauss_constraint_E_K}. This yields
\begin{equation}
\begin{aligned}
   2i \lc{_{k}^{ij}}\tE^a_i\tE^b_j  \curlyD_{[a}K^k_{b]} &= i \lc{_{k}^{ij}}\tE^a_i\tE^b_j \rb{\curlyD_{a}K^k_{b}- \curlyD_{b}K^k_{a}}   \\
   &= \frac{i\kappa}{2\lambda} \rb{-\delta^{il}\tE^a_i\curlyD_aG_l - \delta^{jl}\tE^b_j \curlyD_bG_l}\\
   &=-\frac{i\kappa}{\lambda} \tE^a_j\curlyD_aG^j.
\end{aligned}    
\end{equation}
On the Gau{\ss} constraint surface $G^j$ vanishes everywhere and so does its derivative. Hence -- irrespective of the value of $\lambda$ --  this term does not contribute to the Hamilton constraint. 

We are now left with the contraction of two electric field and two curvature components each. This yields
\begin{equation}
    \lc{_{k}^{ij}}\lc{_{mn}^k}\tE^a_i\tE^b_j K^m_a K^n_b =2 K^{[m}_a K^{n]}_b \tE^a_m\tE^b_n=\rb{(K^m_a \tE^a_m)(K^n_b \tE^b_n)-K^n_a \tE^b_n K^m_b \tE^a_m}.
\end{equation}

Therefore, on the Gau{\ss} constraint surface, the Hamilton constraint in terms of extrinsic curvature and electric field turns into 
\begin{equation}\label{eq:hamilton_constraint_ADM}
    C=\frac{\lambda s_e}{\sqrt{q} \kappa s_E}\rb{qR+\rb{(K^m_a \tE^a_m)(K^n_b \tE^b_n)-K^n_a \tE^b_n K^m_b \tE^a_m}}.
\end{equation}
Again, possible gauge rotations of electric field and extrinsic curvature cancel in the contractions and the Ricci scalar is a function of the spatial metric, which is also invariant under such rotations. 

Recalling the diffeomorphism constraint \eref{eq:diffeo_ADM}, both -- this and the Hamilton constraint \eref{eq:hamilton_constraint_ADM} -- are identical in form to the real ADM constraints. So assuming the reality of electric field and extrinsic curvature, according to the reality conditions, they perfectly match these real constraints up to the factor $\lambda$ and the sign combination $s_e/s_E$.     

The factor $\lambda$ is of course in perfect agreement with the symplectic structure. For $\lambda=1$ we match the real ADM formulation. For $\lambda=i$ the constraints are purely imaginary but so is the Poisson bracket. Hence the described physics does not change. This has therefore the -- rather unintuitive -- interpretation of real ADM gravity that is described by using a purely imaginary action. 

The additional sign $s_e/s_E$ in the Hamilton constraint is the only remnant left from starting with a complex formulation of gravity. This does not change the hypersurface deformation algebra as well and therefore  does not matter.

This connection of the selfdual holomorphic formulation of general relativity to its real ADM formulation ultimately
shows the eligibility of the choice of reality conditions.

\section{Discussion and outlook}\label{sec:discussion}
With the present work we want to contribute to the understanding of general relativity in  Ashtekar's  selfdual variables, as this was the original starting point of loop quantum gravity. To this end, we provided a detailed description of the (anti-)selfdual split of the complex Lorentz Lie algebra and the corresponding complex Lorentz group and its relation to $\SLtwoC$, via the construction of an appropriate isomorphism. 

Special attention was put on working out the consequences of the split for structure constants and the corresponding Cartan-Killing metrics. This, as well as the relation to their $\SLtwoC$ counterpart turned out be of importance when relating the electric field to the spatial metric. 

Starting from the $\laC$-connection of complex Palatini gravity, we further  provided a comprehensive derivation of the $\SLtwoC$ connection, which turns into Ashtekar's connection when pulled back to the spatial manifold. 
In particular, the $\SLtwoC$ formulation follows directly from the $\LGC$ theory. 
Remarkably, this formulation based on the complex Lorentz group needs no extra structures such as a spin structure, in order to go over to the $\SLtwoC$ formulation. 
Lifting our treatment of complex gravity to a bundle formulation would need more work. However, we see no obstruction for doing this at a later point.

Starting from the complex Palatini action, performing its selfdual splitting, restricting to the part of the action in terms of selfdual variables and transferring it to an $\SLtwoC$ formulation via the isomorphism, we gave an extensive derivation of the Hamilton formulation in terms of Ashtekar's selfdual variables. 

In the course of this, two signs, corresponding to the respective orientations of complex tetrads and triads, appear. Their origin is the lack of a definite orientation of complex objects, which we could just require. They do not pose any problem, even when recovering real ADM gravity. As it was not possible to relate the orientation of tetrads and triads in this context, further work is needed in order to understand their role here.  

Different from the treatment of theories with complex variables, but based on real actions, we show in detail that complexified GR based on the complex and in fact holomorphic Palatini action only leads to a non-degenerate symplectic structure that respects this holomorphic viewpoint. Without this, as done in \cite{wieland_complex_2012}, one would have to add the missing Poisson brackets by hand but then in turn would not be working with the complexified Palatini theory anymore. 

Considering the Poisson relation, the phase, which we added in front of the action, became relevant. While the constraint structure is unaffected by this, it allows to consider to special cases. In its absence, i.\,e., for the value $1$, it gives the Poisson relation for selfdual variables as considered in the main part of the literature, cf. \cite{ashtekar_new_1986,baez_gauge_1994,ashtekar_lectures_1998,wilson-ewing_loop_2015,ashtekar_gravitational_2021}. For the value $i$ however, we can enforce the Poisson relation of real Ashtekar-Barbero gravity without the Barbero-Immirzi parameter, but without restricting to a description in terms of real variables. Similar Poisson relations are used in \cite{thiemann_account_1995,thiemann_reality_1996,wieland_complex_2012}. 
Even thought distinguishing between those two Poisson relation has no qualitative impact on the classical formulation, it becomes relevant once its quantisation is considered \cite{sahlmann_revisiting_HS_RCI,sahlmann_revisiting_RCII}.

While the triads of Ashtekar-Barbero gravity   are obtained from the tetrads by fixing time gauge, which causes the internal three metric to be the spatial part of the Minkowski metric, i.\,e., the Euclidean metric\footnote{~See for example \cite{ashtekar_background_2004}.}, the same internal three metric is obtained for  selfdual Ashtekar gravity, but without the necessity of any gauge fixing. 
Assuming a different internal metric is however not compatible with reproducing the ADM formulation under reality conditions. 

In the logic of the applied procedure, however, the complex formulation is only physically useful once the reality conditions are implemented. 
Nevertheless it would be very interesting to look at this complex version of gravity by itself and to understand the physics that might be described by it. 

With this work we further want to provide a detailed resource for the derivations of and considerations about the reality conditions and how they allow to recover real gravity formulations. In addition to that we showed that the addition phase in the formulation does not affect the conditions. 

In future work it would be interesting to consider the inclusion of matter into the selfdual holomorphic gravity formulation. It is expected that matter coupling, as it changes the dynamics in form of the constraints, at least affects the second reality condition. The relation to the symmetry-reduced context will be discussed elsewhere. 

Besides a clear formulation of the classical theory of selfdual Ashtekar gravity, this is supposed to provide the groundwork for a quantum theory of selfdual loop quantum gravity. The quantisation will be considered in \cite{sahlmann_revisiting_HS_RCI,sahlmann_revisiting_RCII}.

\ack
R.S. and H.S. benefitted from discussions with Wojciech Kamiński, Jerzy Lewandowski, Thomas Thiemann, Madhavan Varadarajan and Wolfgang Wieland. H.S. acknowledges the contribution of the COST Action CA18108. R.S. thanks the Evangelisches Studienwerk Villigst for financial support. 

\appendix 

\section{Complexification of Lie algebras}
\label{app:complexification_lie_algebras}
The complexification $\g_\CC$ of a real Lie algebra $\g$ is (ex. \cite{hall_lie_2015})
\begin{equation}    
\g_\CC=\g\otimes_\RR \CC. 
\end{equation}
If $\g\subset M_n(\CC)$ is a matrix Lie algebra such that for any $a\in \g$, $ia$ is not in $\g$, then (ex. \cite{hall_lie_2015} Prop. 3.38)
\begin{equation}
    \g_\CC \cong \{ a \in M_n(\CC): a= x+iy,\; x,y\in \g \}. 
\end{equation}
This is true in particular for Lie algebras of real matrices, and so 
\begin{equation}
    \sothreeone_\CC \cong \sothreeone + i\,\sothreeone  \cong \{a \in M_4(\CC): a\, \eta \, a^T = \eta\}. 
\end{equation}
Also, we have 
\begin{equation}
\label{eq:real_form_sl2C}
    \sutwo_\CC \cong \sltwoC.  
\end{equation}

If $\g$ is already a complex Lie algebra, one can regard it as a real Lie algebra $\g_\RR$ with twice the dimension, and complexify again. One then obtains 
\begin{equation}
    \g_{\RR\CC}=\g\oplus\overline{\g}, 
\end{equation}
with $\overline{\g}$ the conjugate to $\g$ \cite{samelson_notes_1990}. In particular, 
\begin{equation}
    \sltwoC_{\RR\CC} \cong\sltwoC\oplus\overline{\sltwoC}\cong\sltwoC\oplus\sltwoC
\end{equation}
since $\sltwoC$ has real forms (for example $\sutwo$, \eref{eq:real_form_sl2C}), and is thus isomorphic to its conjugate.


\end{document}